\documentclass[times,preprint,12pt]{elsarticleown}

\usepackage{amssymb}
\usepackage{amsthm}

\usepackage{amsmath}
\usepackage{arydshln}
\usepackage[american]{babel}
\usepackage{hyperref}
\usepackage{multirow}
\usepackage{url}
\usepackage{xstring}
\usepackage[margin=2.5cm]{geometry}

\usepackage[figuresright]{rotating}

\newcommand{\abs}{\text{abs}}
\newcommand{\bnal}{\text{NAL}}
\newcommand{\unal}{\text{UNAL}}
\newcommand{\baal}{\text{AAL}}
\newcommand{\arm}{\text{ARM}}
\newcommand{\darm}{\text{DARM}}
\newcommand{\narm}{\text{NARM}}
\newcommand{\unarm}{\text{UNARM}}
\newcommand{\parm}{\text{PARM}}
\newcommand{\exparm}{\text{ExpARM}}
\newcommand{\aarm}{\text{AARM}}

\newcommand{\pdal}{\text{PDAL}}
\newcommand{\expdal}{\text{ExpDAL}}

\newcommand{\paal}{\text{PAAL}}
\newcommand{\dal}{\text{DAL}}
\newcommand{\sat}{\text{3SAT}}

\newcommand{\nspace}{\text{NSPACE}}
\newcommand{\ntime}{\text{NTIME}}

\newcommand{\p}{\text{P}}

\newcommand{\ph}{\text{PH}}
\newcommand{\np}{\text{NP}}

\newcommand{\nl}{\text{NLOGSPACE}}
\newcommand{\pspace}{\text{PSPACE}}

\newcommand{\bigO}{\mathcal{O}}

\newcommand{\grammar}{(N, \Sigma, S, P)}
\newcommand{\polylog}{\mathrm{polylog}}
\newcommand{\icol}[1]{
	\left(\begin{smallmatrix}#1\end{smallmatrix}\right)%
}

\hypersetup{
  linkcolor  = blue,
  citecolor  = blue,
  urlcolor   = blue,
  colorlinks = true,
}

\theoremstyle{plain}
\newtheorem{theorem}{Theorem}[section]
\newtheorem{corollary}[theorem]{Corollary}
\newtheorem{lemma}[theorem]{Lemma}
\newtheorem{hypothesis}[theorem]{Hypothesis}

\theoremstyle{definition}
\newtheorem{definition}[theorem]{Definition}
\newtheorem{example}[theorem]{Example}

\theoremstyle{remark}
\newtheorem*{remark}{Remark}

\begin{document}

\begin{frontmatter}




\title{A Computation Model with Automatic Functions and Relations as Primitive Operations\tnoteref{t}}

\tnotetext[t]{Ziyuan Gao, Sanjay Jain, and Frank Stephan were supported in part by
	the Singapore Ministry of Education Academic Research Fund Tier 2
	grant MOE2019-T2-2-121 / R146-000-304-112; furthermore, Sanjay Jain
	was also supported in part by the NUS Provost Chair grant C252-000-087-001.
}


\author[1]{Ziyuan Gao}
\ead{matgaoz@nus.edu.sg}
\author[2]{Sanjay Jain}
\ead{sanjay@comp.nus.edu.sg}
\author[3]{Zeyong Li}
\ead{li.zeyong@u.nus.edu}
\author[2]{Ammar Fathin Sabili}
\ead{ammar@comp.nus.edu.sg}
\author[1,2]{Frank Stephan}
\ead{fstephan@comp.nus.edu.sg}

\address[1]{Department of Mathematics, National
	University of Singapore, 10 Lower Kent Ridge Road, S17,
	Singapore 119076, Republic of Singapore}
\address[2]{Department of Computer Science,
	National University of Singapore, 13 Computing Drive, COM1,
	Singapore 117417, Republic of Singapore}
\address[3]{Center for Quantum Technologies,
National University of Singapore, 3 Science Drive 2, S15,
Singapore 117543, Republic of Singapore}

\begin{abstract}
Prior work of Hartmanis and Simon \cite{Hartmanis74} and Floyd and Knuth \cite{Floyd90} investigated 
what happens if a device uses primitive steps more natural than single updates of a Turing tape. One finding was that in the numerical setting, addition, subtraction, comparisons and bit-wise Boolean operations of numbers preserve polynomial time while incorporating concatenation or multiplication allows to solve all PSPACE problems in polynomially many steps. Therefore we propose to use updates and comparisons with automatic functions as primitive operations and use constantly many registers; the resulting model covers all primitive operations of Hartmanis and Simon as well as Floyd and Knuth, but the model remains in polynomial time. The present work investigates in particular the deterministic complexity of various natural problems and also gives an overview on the nondeterministic complexity of this model.
\end{abstract}

\begin{keyword}
Theory of Computation \sep Model of Computation \sep Automatic Function \sep Register Machine \sep Computational Complexity



\end{keyword}

\end{frontmatter}



\section{Introduction}
\label{sec:introduction}

\noindent
Sequential models of computation such as finite state machines, register machines
and Turing machines typically carry out primitive operations at each computation
step.  For example, Turing machines have a finite number of rules for 
manipulating symbols on a tape and can only process symbols one at a time.
Register machines, in particular random-access machines (RAM's), can store numbers
of arbitrary size in an unlimited number of registers and have basic 
arithmetic operators such as increment and decrement; a RAM may also be
augmented with instructions for performing advanced operations such as 
multiplication, division, concatenation as well as number comparison 
\cite{Hartmanis74,Schonhage79}.  
A fairly common theme in the study of these machines is the analysis of whether,
and if so to what extent, computational power is boosted by allowing extra
types of instructions.  In the case of RAM's, Hartmanis and Simon \cite{Hartmanis74}
showed that adding logical operations on bit vectors in parallel and either 
multiplication or concatenation allows polynomial time decidability of all languages in \pspace, 
where the instructions use all counts as constant time operations.
Sch\"{o}nhage \cite{Schonhage79} proved that incorporating the arithmetic 
operations of addition, multiplication and division
as primitive steps allows polynomial time decidability of all languages in \np;
if RAM's only incorporate addition and multiplication as primitive steps, then 
the languages they decide in polynomial time are also decidable in polynomial
time by probabilistic Turing machines.  

A natural question is whether the computational power of register machines 
can be enhanced with primitive operators that generalise bit-wise operations, 
addition, subtraction and comparison, but do not lead to an ``unreasonable''
speedup in computational time.  For the purpose of the current work, we
define a ``reasonable'' primitive operator as one that can be simulated by 
an appropriate type of Turing machine or register machine with at most a 
polynomial increase in time complexity.  As mentioned earlier, RAM's that
are equipped with bit operations as well as the multiplication or concatenation 
operator can decide languages in \pspace\ in polynomial time; if one  
allows the addition, multiplication and division operators, then RAM's can 
decide languages in \np\ in polynomial time.  
Motivated by the common assumption that neither \pspace\ nor \np is equal 
to \p, the present work does not regard multiplication or concatenation as a 
reasonable primitive operator. 
In this work, we propose the use of {\em automatic functions}
or, more generally, {\em bounded automatic relations} as primitive steps in a register 
machine.  We call such machines {\em Deterministic Automatic Register Machines}
(\darm 's) in the case that automatic functions are used, and {\em Non-deterministic
Automatic Register Machines} (\narm 's) in the case that bounded automatic relations are
used.  Automatic functions and relations belong to the wider class of 
{\em automatic structures}, which are relational structures whose domain
and atomic relations can be realised by finite automata operating synchronously
on their input.  The idea of using automata to study structures was originally 
conceived by B\"{u}chi, who applied this notion to prove the 
decidability of a fragment of second-order arithmetic known as \textit{S}1\textit{S} 
\cite{Bu60,Bu62}.
The definition of an automatic structure first appeared in Hodgson's doctoral
thesis \cite{Hodgson76}, parts of which were published later \cite{Hodgson82}. 
Khoussainov and Nerode \cite{KN95} rediscovered the concept of an 
automatic structure and defined various notions of automatic presentations. 
They then studied the problem of characterising structures that have automatic 
presentations and analysed the complexity of problems formulated over automatic 
structures.  A large body of work on automatic structures has since emerged,
covering themes ranging from the model checking problem to the problem of 
characterising definable relations \cite{Gradel20,Khoussainov10,Nies07,Rubin08}.

The key link between automatic functions and one-tape
deterministic Turing machines, which forms the basis for the \darm\
model, is that each computation of an automatic function
can be simulated by a one-tape deterministic Turing machine in
linear time \cite{Case13}.  Using a proof technique very similar
to that for establishing the latter result, one can show that each 
computation of a bounded automatic relation can be simulated 
by a one-tape nondeterministic Turing machine in linear time; 
this means there is a one-tape nondeterministic Turing machine 
such that for any pair $(x,y)$ in the automatic relation, there is
a computational path that starts with the input $x$ (with the tape 
head starting at the first position of the input) and eventually
reaches an accept state with the final tape content $y$, where
the first position of $y$ matches that of the input.  Thus
a computation by an automatic function or bounded automatic relation 
fulfills our criterion for a reasonable primitive step.
In particular, while automatic functions can perform bit-wise operations, 
addition, subtraction and comparison,
they are not powerful enough to carry out multiplication or concatenation.

Automatic register machines (\arm 's) may be viewed as a generalisation of
{\em addition machines}, a class of register machines introduced by
Floyd and Knuth \cite{Floyd90}.  Addition machines are limited to 
six types of operations, each carried out in constant time: reading input, 
writing output, adding, subtracting, copying the contents of one register 
to another register and comparing the contents of any two registers.
It turns out that many natural operations like multiplication can be done 
by addition machines in linear time \cite{Floyd90}. These nice results 
carry over to automatic function complexity, as the natural numbers with addition 
and comparison is an automatic structure.  Moreover, as mentioned earlier,
an \arm\ can also perform each bit-wise Boolean operation in a single step
\cite{Hartmanis74,Hartmanis76}.
Another class of devices closely related to \arm 's is that of the {\em finite-state
transducers}.  It may seem reasonable to use the class of {\em rational relations} -- 
relations recognised by finite-state transducers -- as a primitive operator in 
register machines.
However, as mentioned earlier, 
Hartmanis and Simon \cite[Theorem 2]{Hartmanis74} showed that
bitwise operations together with addition, subtraction, comparisons and 
concatenation allow to decide in polynomially many steps sets from \pspace; 
as all these primitive functions and relations can be implemented using transducers, it follows that transducers as primitive steps are too powerful, see Section \ref{subsec:transducer_pspace} 
for more explanations and related results.  For this reason, we have chosen not 
to implement rational relations as primitive operators.  

One might also propose combining register machines with transducers 
that are limited by the constraint that for every update function which is computed 
by the transducer, there is a constant $c$ such that for all inputs $x$ the 
output $y$ is at most $c$ symbols longer than $x$;
such a transducer would therefore satisfy a size constraint property
similar to that of the automatic functions.  But such transducers might
still be regarded as being too powerful, for they allow nonregular sets
to be recognised in constantly many steps.  One such example is the
set of all words that have as many $0$'s as $1$'s:
for each $a,b \in \{0,1\}$, let $f_{a,b}(x)$ be the output of a transducer
$T$ which reads a word $x$ and outputs a symbol $b$ whenever it reads the symbol
$a$ and which discards all other symbols without replacement;
$T$, which evidently satisfies the required size constraint property,
accepts an input $x$ if and only if $f_{0,1}(x) = f_{1,1}(x)$.
On the other hand, as we will show in this work, the class of languages recognised
in constantly many steps by \darm 's belongs to the class of regular languages,
see Theorem \ref{the:darm_regular} below.  
Thus, the computational power of register machines with primitive steps computed by 
transducers whose output has length bounded by the input length plus an additive constant
already exceeds that of register machines combined with automatic functions.  

An alternative model of transducers that can indeed be simulated by deterministic Turing
machines in polynomial time is that recently proposed by Kutrib, Malcher, Mereghetti, and Palano \cite{Kutrib20}, where transducers are restricted to be length-preserving.
Although this transducer model yields interesting results, it might not be adequate
for modelling primitive steps in a register machine since length-preserving transductions are a 
real restriction and allow recognition of only a subclass of \p.  Automatic functions are more 
robust in that they can output a string longer than the input by an additive constant, and cover
the full class \p\ in polynomially many steps.  For all these reasons, we believe that
automatic functions and bounded automatic relations
are an appropriate choice of primitive operators in register machines.     

We also mention a few other models of computation whose computational
powers appear to be similar to that of \arm 's.  First, a {\em $d$-dimensional
array automaton} consists of an 1-way input tape and a $d$-dimensional
regular array of cells, each of which contains a finite-state machine, and
each step of computation consists of an input head move and an array head move
together with a state transformation for each finite-state machine in a cell.
Kosaraju \cite{Kosaraju75} showed that context-free languages can be recognised
by 1-dimensional arrays in $n^2$ time; using a similar algorithm, we show
the same time complexity can be achieved by a deterministic \arm\ (Theorem \ref{the:cfg_dal_n2}).
Kutrib and Malcher \cite{Kutrib09,Kutrib13} also investigated one-dimensional
arrays of interconnecting finite automata (known as {\em cellular automata} or {\em cells}),
showing in particular that cellular automaton transducers enjoy better time complexities
compared to iterative array transducers, where the input is processed in real time
and the output is computed in linear time.  We note that any computation of an
automatic function can be simulated by one computation step of a cellular automaton;
it would therefore appear that cellular automata are just as powerful as deterministic
\arm 's.   

The contributions of this work are as follows.  First, we lay out the general definition
of an \arm, and then define two specific types of such machines, namely 
\darm 's and \narm 's, as well as their respective 
complexity classes for arbitrary functions.  Next, we compare the computational
power of \arm 's to that of other models of computation including transducers,
Generative systems (g-systems) \cite{Rovan81} and Iterated Uniform Finite-State
Transducers \cite{Kutrib20}.  The main goal of these comparisons is to justify the use 
of automatic functions and bounded automatic relations in combination with register 
machines; we wish to show, in particular, that it yields a sufficiently expressive model 
that does not give an unrealistic speed-up.          
Furthermore, we show that an open question of Kutrib, Malcher, Mereghetti, and 
Palano \cite{Kutrib20}
on whether their model of Iterated Uniform Finite-State Transducers covers the whole 
class of context-sensitive languages has a positive answer; this can be shown using a 
proof quite similar to that for establishing that
context-sensitive languages are in nondeterministic linear space.

We then proceed to the main results concerning \darm 's,
first establishing that those languages recognised by such machines
in constant time and polynomial time coincide with the regular languages
and \p\ respectively.  It is shown that various problems such as
recognising context-free languages, recognising Boolean languages and
the multi-sources connectivity problem have a lower time complexity in the 
\darm\ model in comparison to plain RAM models.  We obtain almost
tight complexity bounds for some of these problems.  Owing to the 
tight connection between automatic functions and deterministic 
one-tape Turing machines, traditional crossing sequence arguments
used to establish lower bounds on the running time of one-tape
Turing machines apply to \darm 's as well.  

The paper then
moves on to \narm 's, in particular showing that the class of languages recognised
by \narm 's in polynomial time coincides with \np.  The time complexity of
various problems in the \narm\ model, including that of recognising non-palindromes,
\sat\ and context-free languages are then studied.  The key observation
for results concerning \narm 's is similar to that in the case of \darm 's, that is, 
one often gets a speedup in time complexity that is nonetheless not too unreasonable.

We conclude by examining more powerful variants of \arm 's -- first, by allowing
unbounded automatic relations; second, by padding the input at the start of
the computation with ``working space'' that is larger than the original input 
length.  In the second variant, if the input is padded with a string that is
at least polynomially longer than the original input length, then one obtains
a Polynomial-Size Padded Automatic Register Machine (\parm), and if
the padding string is at least exponentially longer than the original input
length, then an Exponential-Size Padded Automatic Register Machine (\exparm)
is obtained.  The increase in computational power as a result of padding
the input is illustrated by the fact that QSAT -- the Quantified Boolean Formula 
Satisfaction problem -- can be decided by an \exparm\ in linear time (as we will 
show in Section \ref{sec:other_arm}), as well as the fact that 
the class of languages decidable in polynomial time by \exparm 's is equal
to \pspace.  Moreover, even a polynomial-size padding can boost 
computational power: one can show that a $(c\log(n))$-variable version of 
3SAT can be decided by a \parm\ in polylogarithmic time.
Finally, we prove that under the Exponential Time Hypothesis (ETH),
the class of languages decided by \parm 's in sublinear time is strictly 
greater than that decided by \darm 's in sublinear time.  
Many results in this paper build on prior work on automatic
functions and relations \cite{Case13,zeyongli,Seah}. 

\section{Preliminaries}
\label{sec:preliminaries}

\noindent
Let $\Sigma$ denote a finite alphabet. We consider set operations including union ($\cup$), concatenation ($\cdot$), Kleene star ($\ast$), intersection ($\cap$), and complement ($\lnot$). Let $\Sigma^*$ denote the set of all strings over $\Sigma$. Let the empty string be denoted by $\varepsilon$. For a string $w \in \Sigma^*$, let $|w|$ denote the length of $w$ and $w = a_1a_2...a_{|w|}$ where $a_i \in \Sigma$ denotes the $i$-th symbol of $w$. We also use $w_{i,j}$ to denote the substring $a_ia_{i+1}...a_j$. Fix a special padding symbol $\#$ not in $\Sigma$. Let $x, y \in \Sigma^*$ such that $x = x_1x_2\ldots x_m$ and $y = y_1y_2\ldots y_n$. Let $x' = x_1'x_2'\ldots x_r'$ and $y' = y_1'y_2'\ldots y_r'$ where $r = \max(m,n)$, $x_i' = x_i$ if $i \leq m$ else $\#$, and $y_i' = y_i$ if $i \leq n$ else $\#$. Then, the convolution of $x$ and $y$ is $conv(x,y) = \binom{x_1'}{y_1'} \binom{x_2'}{y_2'} \ldots \binom{x_r'}{y_r'}$ \cite{Seah}. In other words, a convolution of two (or more) strings is formed by making pairs symbol-by-symbol with appropriate padding at the end for shorter strings. A relation $R \subseteq X \times Y$ is automatic iff the set $\{conv(x,y): (x,y) \in R\}$ is regular. Likewise, a function $f : X \rightarrow Y$ is automatic iff the relation $\{(x,y) : x \in domain(f) \land y = f(x)\}$ is automatic  \cite{Stephan}. An automatic relation $R$ is bounded iff $\exists$ constant $c$ such that $\forall (x,y) \in R,~\abs(|y| - |x|) \leq c$. On the other hand, an unbounded automatic relation has no such restriction. Automatic functions and relations have a particularly nice feature as shown in the following theorem.

\begin{theorem}[Hodgson \cite{Hodgson83}, Khoussainov and Nerode \cite{KN95}]\label{the:first_order_definable}
	Every function or relation which is first-order definable from a finite number of automatic functions and relations is automatic, and the corresponding automaton can be effectively computed from the given automata.
\end{theorem}

\noindent
A grammar is a 4-tuple $\grammar$ where $N$ and $\Sigma$ are disjoint sets of non-terminal and terminal symbols. $S \in N$ is the starting symbol and $P$ is the set of rules. Different restrictions on the rules in $P$ define different family of languages. Moreover, $l \Rightarrow r$ denotes that $r$ can be derived from $l$ in one step while $l \Rightarrow^* r$ denotes that $r$ can be derived from $l$ in zero or more steps. A context-free grammar $\grammar$ is a grammar in which every rule in $P$ is restricted to be in the form of $A \rightarrow \alpha$ where $A \in N$ and $\alpha \in (N\cup\Sigma)^*$. A context-free language is a language generated by a context-free grammar. For simplicity, we assume that the language does not contain the empty string. A context-free grammar $\grammar$ is in Chomsky Normal Form if every rule in $P$ is restricted either in the form of $A \rightarrow u$ or $A \rightarrow BC$ where $A,B,C \in N$ and $u \in \Sigma$. A context-free grammar $\grammar$ is in Greibach Normal Form if every rule in $P$ is restricted in the form of $A \rightarrow \alpha$ where $A \in N$ and $\alpha \in \Sigma N^*$. In informal definition by Okhotin \cite{Okhotin13}, a Boolean grammar $\grammar$ extends context-free grammars by allowing two additional types of rules in $P$: (1) $A \rightarrow \alpha ~\&~ \beta$ where $A \in N$ and $\alpha, \beta \in (N \cup \Sigma)^*$: this type of rule states that for any $w \in \Sigma^*, A \Rightarrow^* w$ if $ \alpha \Rightarrow^* w$ and  $\beta \Rightarrow^* w$; and (2) $A \rightarrow \lnot \alpha$ where $A \in N$ and $\alpha \in (N \cup \Sigma)^*$: this type of rule states that for any $w \in \Sigma^*, A \Rightarrow^* w$ if $ \alpha \not\Rightarrow^* w$. A Boolean language is a language generated by a Boolean grammar. A Boolean grammar $\grammar$ is in Binary Normal Form if every rule in $P$ except $S \rightarrow \varepsilon$ is restricted either in the form of $A \rightarrow u$ or $A \rightarrow B_1C_1 ~\&~ \cdots ~\&~ B_mC_m ~\&~ \lnot D_1E_1 ~\&~ \cdots ~\&~ \lnot D_nE_n ~\&~ \lnot \varepsilon ~ (m \geq 1, n \geq 0)$ where $A,B_i,C_i,D_i,E_i \in N$ and $u \in \Sigma$.

This paper uses a standard definition of Turing machines by Arora and Barak \cite{CC} and focuses on single-tape Turing machines. We consider a one-dimensional tape that is infinite to the right, consisting of a sequence of cells. Each cell can hold a symbol from a finite set $\Gamma$ called the alphabet of the machine. Note that the alphabet of a language $\Sigma$ will be 
a subset of $\Gamma$ in this Turing machine model. The whole (non-empty) content of the tape at one point of time is defined as the content of the machine. The tape is equipped with a tape head pointing a cell to read or write the symbol on that cell in a single step. The tape head is initially pointing to the leftmost cell and can move left or right in a single step. A Turing machine $M$ is described as a triple $\langle\Gamma, Q, \delta\rangle$ containing:

\begin{enumerate}
	\item The alphabet of the machine $\Gamma$. Let $\square \in \Gamma$ be a symbol denoting an empty cell.
	\item A set $Q$ consisting of the possible states of the machine $M$. Let $q_{start}$, $q_{halt\_accept}$, $q_{halt\_reject} \in Q$ denote the start state, halt then accept state, and halt then reject state respectively.
	\item A transition function $\delta : Q \times \Gamma \rightarrow Q \times \Gamma \times \{L,R,S\}$ describing the rules of $M$ in performing each step. The rule describes: if the machine is in some state $\in Q$ and the tape head reads some symbol $\in \Gamma$ then the machine will change its state to some state $\in Q$, the tape head will write some symbol $\in \Gamma$ in the cell it is pointing to, and then it either moves left, right, or stay.
\end{enumerate}

\noindent
An input is initially written on the leftmost cells of the tape and other cells will be empty. The time complexity of a machine on input $x$ is the minimum number of steps needed to reach $q_{halt,accept}$ or $q_{halt,reject}$. 

The exact concept of register machine varies from author to author. This paper uses a definition of register machine which only has constantly many registers --- each capable of storing a non-negative integer --- therefore the indirect access operation $(R_i \leftarrow R_{R_j})$ is not used. We use basic arithmetic operations and comparisons $(<,>,=)$ as in Floyd and Knuth \cite{Floyd90} in addition to allowing constant values assignment and bitwise operations (as found in various programming languages like C, C++, and Javascript) by Hartmanis and Simon \cite{Hartmanis74}. Hartmanis and Simon observed that including either concatenation or multiplication makes the concept to cover $PSPACE$ in polynomially many steps therefore both are not used by us. Furthermore, inputs were read explicitly in the program and writing was only possible in programs which compute functions. The statements are numbered (with line numbers) and conditional and unconditional $goto$ are allowed. For computing characteristic functions of set, the explicit $halt\_accept$ and $halt\_reject$ operations are used. Readers are advised to check Cook and Reckhow \cite{Cook73} for the relation between register machine (using a RAM model) and Turing machine.

\section{Automatic Register Machines}
\label{sec:model}

\subsection{Formal Definition}
\label{subsec:arm}

\noindent
Formally, we denote an Automatic Register Machine ($\arm$) by $M$ and represent it as a quintuple $\langle \Gamma,\Sigma,R,Op,P \rangle$ by the following definition. Semantically, our model tries to generalize the standard register machines where the operations are no longer limited and the registers are able to store string. The operations, however, are restricted depending on the the type of $M$: in the most basic type, Deterministic Automatic Register Machine ($\darm$) only allows $Op$ to be the set of automatic functions where the parameters (and the output) are convoluted. Note that in standard register machine, all operations are automatic, thus register machine is a subset of $\darm$. More types of $M$, which are also more general than $\darm$, will be discussed later in Subsection \ref{subsec:arm_turing} and Section \ref{sec:other_arm}.

\begin{definition}[Automatic Register Machine, Seah \cite{Seah}]
	
	Automatic Register Machine ($\arm$) is a machine $M = \langle \Gamma,\Sigma,R,Op,P \rangle$ where $M$ is specified by:
	
	\begin{enumerate}
		\item a finite set of registers $R = \{r_1, r_2, \cdots, r_n\}$, each capable of storing a string, possibly empty, over a fixed machine alphabet $\Gamma$;
		\item a finite set of possible operations $Op = \{f_1, f_2, \cdots, f_m\}$, where the $f_i$'s are automatic relations and are further restricted to some rules depending on the type of $M$ (for example, the $f_i$'s are automatic functions if $M$ is deterministic and are bounded automatic relations if $M$ is nondeterministic);
		\item a program $P$ consisting of a finite list of instructions in form of $LINE\_NO : INSTR$. \newline The $INSTR$ takes one of the following forms:
		\begin{itemize}
			\item $read(R_i)$: assigning the input string over a fixed input alphabet $\Sigma \subseteq \Gamma$ to register $R_i \in R$,
			\item $write(R_i)$: writing the output string of register $R_i \in R$,
			\item $R_i \leftarrow w$: assigning any string $w$ over a fixed input alphabet $\Sigma \subseteq \Gamma$ to register $R_i \in R$,
			\item $R_i \leftarrow R_j$: assigning the value of register $R_j \in R$ to register $R_i \in R$,
			\item $R_i \leftarrow F(R_1, R_2, \cdots, R_k)$: assigning the result of operation $F \in Op$ with parameters $R_1, R_2, \cdots, R_k \in R$ to register $R_i \in R$,
			\item $goto(LINE\_NO)$: go to the instruction in line $LINE\_NO$ (i.e. unconditional jump),
			\item $if~F(R_1, R_2, \cdots, R_k)~then~goto(LINE\_NO)$: if the result of operation $F \in Op$ with parameters $R_1, \cdots, R_k \\ \in R$ is true, then go to the instruction in line $LINE\_NO$ (i.e. conditional jump);
			\item $halt$ (or $halt\_accept$ or $halt\_reject$): halting the program (and accept or reject) the input.
		\end{itemize}
	\end{enumerate}
	
\end{definition}

\noindent
Similar to register machine, the program runs sequentially from the first line number to the next one unless jumped either conditionally or unconditionally, while performing the instructions along the runtime. A single step is defined as performing one line of instruction. Initially, all registers store empty string. After several steps, the program terminates after halted by $halt$. If one wants to recognize a language, both halt states (i.e. $halt\_accept$ or $halt\_reject$) can be used. Alternatively, one can in deterministic model write a bit output of `1'/`0' then halt to denote the acceptance/rejection. If one wants to write a function value as an output, both halt states are also allowed especially for nondeterministic models. This is because the whole computation can be rejected when it turns out in guess-then-verify steps that some info had been guessed wrongly and thus leads to rejection of the computation. A partial constant-valued function can also be done in this way where the output might be skipped (by $halt\_reject$) as only the output on members of some language $L$ is considered.

\begin{remark}
	Note that reading in a register will erase the previous content of the register. Thus if you have a register machine which sums up all integer inputs until it reads $0$, then this register machine needs at least two registers, as otherwise the reading of a new input would erase the old content. One register is only possible if one reads in the first statement a convolution of all inputs (or in other longer format like the concatenation of all inputs).
\end{remark}

\begin{example}
	\label{exa:l_two}
	
	We give an illustration of an Automatic Register Machine $M$ recognizing $L_{two} = \{0^{2^i}$ for some non-negative integer $i\}$. Let $M = \langle \Gamma,\Sigma,R,Op,P \rangle$, $\Gamma = \{0,1\}$, $\Sigma = \{0\}$, $R = \{r_1, r_2\}$, $Op = \{is\_empty, is\_one, is\_odd,divide\_by\_two\}$, and
	
	\item \begin{align*}
		P = \{~~
		01 &: read(r_1) \\
		02 &: if~is\_empty(r_1)~then~goto(09) \\
		03 &: if~is\_one(r_1)~then~goto(08) \\
		04 &: if~is\_odd(r_1)~then~goto(09) \\
		05 &: r_2 \leftarrow divide\_by\_two(r_1) \\
		06 &: r_1 \leftarrow r_2 \\
		07 &: goto(03) \\
		08 &: halt\_accept \\
		09 &: halt\_reject
		~~\}
	\end{align*}
	
	\noindent
	with the following definition of the operations:
	\begin{itemize}
		\item $is\_empty(w)$ returns true iff $w$ is an empty string,
		\item $is\_one(w)$ returns true iff $w$ has exactly one occurrence of symbol `0',
		\item $is\_odd(w)$ returns true iff $w$ has odd occurrences of symbol `0',
		\item $divide\_by\_two(w)$ returns $w'$ where $w'$ is equal to $w$ after modifying every second occurrence of `0' to `1' (e.g. if $w = 10\underline{0}1011\underline{0}$ then $w'=10\underline{1}1011\underline{1}$.)
	\end{itemize}
	
	\noindent
	It is not hard to see that each operations can be constructed with finite automatons thus they are automatic. Note that for $divide\_by\_two$ in particular, the input $w$ and the output $w'$ are convoluted. As $Op$ only consists of automatic functions, therefore $M$ satisfies the requirement as a DARM. Let $n$ be the occurrence number of `0' in $r_1$, which is initially equal to the length of the input string. The proof of correctness of the program comes from keep dividing $n$ by $2$ until $n$ is either $1$ --- then accept because $n$ is $2^i$ for some non-negative integer $i$ --- or other odd number --- then reject. The program will halt and either accept or reject the input in no more than $O(\log n)$ steps.
	
\end{example}

\subsection{Alternative Representation with Single Operation}
\label{subsec:single_register}

\noindent
The representation of an ARM $M = \langle \Gamma,\Sigma,R,Op,P \rangle$ could be simplified by rewriting each instructions in $P$ to a form of $LINE\_NO : F(R_1, R_2, \cdots, R_k)$ for some operations $F$ and $R_1, R_2, \cdots, R_k \in R$. The operation $F$ may be a member of $Op$ or also a ``keyword" operation (e.g. $read$, $goto$, $halt$). This new form resembles a neater operation format in register machine such as $INC(R_i)$, $ADD(R_i,R_j)$, and $IFEQUAL(R_i,R_j,LINE\_NO)$; which are all automatic functions. We can then simplify the machine representation even more: the $LINE\_NO$ could also be inserted as a parameter in $F$ as well as all unused registers --- which are eventually ignored when the operation is performed. Thus, the instructions in $P$ could be fully rewritten as $F(LINE\_NO, r_1, r_2, \cdots, r_n)$ where $F$ is either a member of $Op$ or keyword operations.

Moreover, based on Theorem \ref{the:first_order_definable}, the first-order definable function from automatic functions is also automatic. As members of $Op$ and keyword operations are all automatic, thus we can combine them as a single operation $G(LINE\_NO, r_1, r_2,\cdots, r_n)$ where $G$ will run the specific operations based on the $LINE\_NO$. After one step, $G$ will update the new $LINE\_NO$ as well as some registers. The program $P$ then has this single operation $G$ alone which will be applied iteratively from $G(01, \epsilon, \epsilon, \cdots)$, where $01$ denotes the first line number, until it halts. Therefore, we can replace $P$ by $G$ and finally ARM can be re-represented as a triple $M = \langle \Gamma,\Sigma,G \rangle$. One single step is redefined as applying the operation $G$ once.

\begin{definition}[Alternative Representation of Automatic Register Machine]
	An Automatic Register Machine $M = \langle \Gamma,\Sigma,R,Op,P \rangle$ can also be re-represented as $M = \langle \Gamma,\Sigma,G \rangle$ where $G$ is a single operation combining all instructions in $P$ with a convolution of the line number and all registers in $R$ as its parameter.
\end{definition}

\begin{example}
	
	We can rewrite the ARM in Example \ref{exa:l_two} recognizing $L_{two}$ by letting $M = \langle \Gamma,\Sigma,G \rangle$, $\Gamma = \{0,1\}$, $\Sigma = \{0\}$, and $G$ with the following definition:
	
	\begin{equation*}
		G(LINE\_NO,r_1,r_2) = \begin{cases}
			G(02,input,r_2) &\text{if $LINE\_NO = 01$} \\
			G(09,r_1,r_2) &\text{if $LINE\_NO = 02$ and $r_1 = \epsilon$} \\
			G(03,r_1,r_2) &\text{if $LINE\_NO = 02$ and $r_1 \neq \epsilon$} \\
			G(08,r_1,r_2) &\text{if $LINE\_NO = 03$ and `0' occurs $1$ time in $r_1$} \\
			G(04,r_1,r_2) &\text{if $LINE\_NO = 03$ and `0' does not occur $1$ time in $r_1$} \\
			G(09,r_1,r_2) &\text{if $LINE\_NO = 04$ and `0' occurs odd times in $r_1$} \\
			G(05,r_1,r_2) &\text{if $LINE\_NO = 04$ and `0' occurs even times in $r_1$} \\
			G(06,r_1,div2(r_1)) &\text{if $LINE\_NO = 05$} \\
			G(07,r_2,r_2) &\text{if $LINE\_NO = 06$} \\
			G(03,r_1,r_2) &\text{if $LINE\_NO = 07$} \\
			halt\_accept &\text{if $LINE\_NO = 08$} \\
			halt\_reject &\text{if $LINE\_NO = 09$}
		\end{cases}
	\end{equation*}
	
	\noindent
	where $input$ is an input string and $div2(w)$ is the result of modifying 
$w$ 
by changing every second occurrence of `0' to `1'. Note that $G$ is an automatic function. The machine will start with $G(01, \epsilon, \epsilon)$ and eventually will halt in the same number of steps as the original machine which is in $\bigO(\log n)$ steps where $n$ is the length of the input.
	
\end{example}

\begin{remark}
	Note that it is also possible to represent all registers in a single (big) register which is simply a convolution of all registers. Moreover, while both representations of $ARM$ are equivalent, the original representation is more natural and mainly used for the sake of readability. However, the alternative representation may also be used when proving some findings especially when they are related to Turing machine computation.
\end{remark}

\subsection{Basic Types of ARM and Relations to Turing machine}
\label{subsec:arm_turing}

\noindent
We now first define the two basic types of $ARM$ which are our main focus throughout this paper. Other types will be discussed later in Section \ref{sec:other_arm}.

\begin{definition}[Deterministic Automatic Register Machine]
	A Deterministic Automatic Register Machine ($\darm$) is an Automatic Register Machine $M = \langle \Gamma,\Sigma,R,Op,P \rangle$ where $Op$ is restricted to automatic functions. Alternatively, a $\darm$ is an Automatic Register Machine $M = \langle \Gamma,\Sigma,G \rangle$ where $G$ is an automatic function.
\end{definition}

\begin{definition}[Nondeterministic Automatic Register Machine]
	A Nondeterministic Automatic Register Machine ($\narm$) is an Automatic Register Machine $M = \langle \Gamma,\Sigma,R,Op,P \rangle$ where $Op$ is restricted to bounded automatic relations. Alternatively, a $\narm$ is an Automatic Register Machine $M = \langle \Gamma,\Sigma,G \rangle$ where $G$ is a bounded automatic relation.
\end{definition}

\noindent
In $\narm$, the bounded version of automatic relations is preferred to avoid the ability to guess a huge witness in one step by unboundedness, which may recognize a larger complexity class as discussed in Subsection \ref{subsec:unarm}. Furthermore, the two concepts differ when act as functions: a nondeterministic bounded automatic function can only have polynomial-sized output, while for the unbounded one it can have exponential-sized output e.g. a concatenation of all $n$-digits binary numbers separated by symbol `2'. We may now define the complexity class of the models.

\begin{definition}[Deterministic Automatic Register Machine Complexity]
	Let $f(n)$ be any function in the input size $n$. DAL[$f(n)$] is the class of languages accepted by $\darm$ in $\bigO(f(n))$ steps.
\end{definition}

\begin{definition}[Nondeterministic Automatic Register Machine Complexity]
	Let $f(n)$ be any function in the input size $n$. NAL[$f(n)$] is the class of languages accepted by $\narm$ in no more than $\bigO(f(n))$ steps.
\end{definition}

\noindent
By this definition, $L_{two}$ in Example \ref{exa:l_two} is in $\dal[\log n]$. Similar to nondeterministic Turing machine, a string in $\narm$ is accepted iff at least one of the nondeterministic computation path halt with accept. Note that all accepting paths must be done in no more than $\bigO(f(n))$ steps for the language being in $\bnal[f(n)]$. Next, the following two theorems give the basic connection between automatic functions and bounded relations to Turing machine computations.

\begin{theorem}(Case, Jain, Seah and Stephan {\cite[Theorem 2.4]{Case13}})\label{the:turing_simulation_deterministic}
A function is automatic iff it can be computed in $\bigO(n)$ time by a deterministic one-tape Turing machine whose input and output start at the same cell, where $n$ is the length of the input string.
\end{theorem}

\begin{corollary}\label{the:turing_simulation}
	Any step by a bounded automatic relation $R$ can be simulated by a one-tape nondeterministic Turing machine in $\bigO(m)$ steps where $m$ is the length of the longest current register content and $m$ is bounded by $n + \bigO(t)$ where $t$ is the number of Automatic Register Machine steps done after the machine reads the input of length $n$.
\end{corollary}

\noindent
The core consequence of Theorem \ref{the:turing_simulation_deterministic} is that we can translate a $\darm$ to an equivalent deterministic Turing machine which runs linear-time longer. Note that the linear here is subjected to the length of the register. To show this, consider any $\darm~M = \langle \Gamma,\Sigma,G \rangle$. The parameters of $G$ are convoluted and this string could be pasted into a Turing tape. In a single step of $M$, the parameters are updated by applying an automatic function $G$ on it. By the theorem, this can be simulated by the Turing machine in linear step of the length of the parameters, thus Q.E.D. Moreover, it is also possible to have the following converse: we can translate a deterministic Turing machine $M$ which runs in $\bigO(f(n))$ to an equivalent $\darm$ which runs in $\dal[f(n)]$. The proof is quite similar where now the content of the Turing tape is pasted as a single register in $\darm$ and show that each Turing machine step is an automatic function. Furthermore, the results can extended to $\narm$ (with nondeterministic Turing Machine) as a consequence of Theorem \ref{the:turing_simulation} in similar manner. Formal proof of these results can be found later in Section \ref{sec:darm} and \ref{sec:narm}.

\section{ARM In Comparison to Other Models}
\label{sec:arm_comparison}

\noindent
This section is solely dedicated to justify that Automatic Register Machine is an adequate computation model which has primitive operations that are expressive and powerful yet not giving an unrealistic speed-up compared to other models.

\subsection{Deterministic Polynomial Time Model with Transducers as Primitive Operations to Solve PSPACE-Complete}
\label{subsec:transducer_pspace}

\noindent
The computation model with transducers as primitive operations is similar to ARM with the only difference that for $M = \langle \Gamma,\Sigma,G \rangle$, $G$ is a transducer. Informally, transducers allow one of the symbol to be $\epsilon$ when ``convoluted'' and such ``convolution'' can be recognized by an automaton (called Mealy machine), thus it is a superset of automatic functions and relations. This model is close to a language generating device called an Iterated Sequential Transducer by Bordihn, Fernau, Holzer, Manca and Mart{\'\i}n-Vide \cite{Bordihn06} albeit the complexity is usually measured by its number of states (i.e. state complexity).  

We show that the computation model with transducers as primitive operations gives an extreme speed-up by solving a $\pspace$-complete problem in only polynomial number of steps. This 
can be deduced from our result in Theorem \ref{the:qsat_expdal} later that a DARM with an access to exponential-length padding (i.e. allowing its working space to be exponential larger) can solve $QSAT$ in $\bigO(n)$ steps, and one can simulate the exponential padding by a transducer which runs through a loop $p(n)$ times and each time doubles the length of $y$ by mapping $y$ to a string twice the length (to get a padding of length at least $2^{p(n)}$). Hartmanis and Simon \cite{Hartmanis74} also show that $NPSPACE \subseteq$ polynomial steps on RAM with concatenation. They write later that they need only bitwise operations, concatenation, and division by 2; not even substring operation. All these operations can be realised by transducers, therefore their result also shows that polynomial number of transducer steps contains $\pspace$. Though the exact construction is a bit technical, the result is not that important in the context of this paper therefore the working out of details and the verification is left to the reader. As a comparison, such blow-up will be unlikely to happen in $\darm$ as based on Theorem \ref{the:dal_ptime} later, if $\sat$ or $QSAT$ are in $\dal[poly(n)]$ then they are also in $P$.

\subsection{ARM and G-systems}

\noindent
Generative systems (g-systems), introduced by Rovan \cite{Rovan81} (stemming from $\Gamma$-systems by Wood \cite{Wood1976}) as a natural model for language generating devices, depend on various types of transducers. We first define the g-systems.

\begin{definition}[Generative systems \cite{Rovan81}]\label{defn:gsystem}
	A generative system (g-system) is a 4-tuple $G=\langle N,\Sigma,M$ $,S \rangle$ where $N$ and $\Sigma$ are disjoint set of non-terminal and terminal symbols, $S \in N$ is the starting symbol and $M$ is a transducer with $M(w) = \emptyset$ for each $w \in \Sigma^+$.
	$M$ is a 6-tuple $(K,N \cup \Sigma, N \cup \Sigma, H, q_0, F)$, where $K$ is a finite set of states, $q_0$ is the initial state,
	$F \subseteq K$ is a set of accepting states and $H$ is a finite subset of $K \times (N \cup \Sigma) \times (N \cup \Sigma)^+ \times K$.
	A computation of $M$ is a word $h_1\ldots h_n \in H^+$ such that if, for each $i \in \{1,2,3,4\}$, $p_i$ is the homomorphism
	on $H^*$ for which $p_i(x_1,x_2,x_3,x_4) = x_i$, then $p_1(h_1) = q_0$, $p_4(h_n) \in F$ and $p_1(h_{i+1}) = p_4(h_i)$
	for $1 \leq i \leq n-1$.  A computation of $M$ of length $n$ may be interpreted as a sequence of $n$ configurations 
	such that the end state of each of the first $n-1$ configurations is the start state of its succeeding configuration;
	furthermore, the first configuration starts with the initial state $q_0$ and the end state of the last 
	configuration is an accepting state.  Let $\prod_M$ denote the set of all computations of $M$.  
	For each language $L \subseteq (N \cup \Sigma)^*$, $M(L)$ is defined to be $p_3\left(p_2^{-1}(L) \cap \prod_M\right)$;
	in other words, for every $w \in L$, if $(x_1^1,x_2^1,x_3^1,x_4^1)(x_1^2,x_2^2,x_3^2,x_4^2)\ldots(x_1^n,x_2^n,x_3^n,x_4^n)$ is a computation such that $x_2^1x_2^2\ldots x_2^n = w$, then
	$x_3^1x_3^2\ldots x_3^n \in M(L)$.            
\end{definition}

As is usually defined for grammars, the language generated by a g-system $G$ consists of all $w \in \Sigma^*$ such that there is a finite sequence $S = v_1,v_2,\ldots,v_k = w$ of words over $N \cup \Sigma$ with $v_{i+1} \in M(\{v_i\})$ for all $i \in \{1,\ldots,k-1\}$.     
We note that the transducer $M$ in Definition \ref{defn:gsystem} maps subsets of $(N \cup \Sigma)^*$ to subsets of $(N \cup \Sigma)^*$ and is thus nondeterministic (otherwise, the language generated by $G$ consists of at most one word).

As a generating language device, one step of derivation in the g-system is not using a set of rules $P$ as in grammar but the transducer $M$. One step of $l \Rightarrow r$ is valid if and only if $r \in M(l)$. Various families of g-systems are differentiated by different types of their transducers. If the transducer does not allow deletion, then it is called $\epsilon$-free g-system, else we can call it a g-system with deletion. If the transducer is some type of directed acyclic graph and either one goes forward in it or one stays on some state and copies and pastes the current symbols, then it is called a sequential g-system, else it is called a parallel g-system. It is not hard to see the similarity between g-system if used as a language accepting device with our model apart from the different types of operations. We list out the connection between them as follows.

\begin{enumerate}
	\item If an $\epsilon$-free g-system accepting (or generating) $L$, there exists a nondeterministic automatic function which also recognizes $L$ such that if the g-system generates some word $w$ in $m$ steps then $w$ can be accepted in $\bnal[m]$ for some $m$. However, if $\narm$ measures the time complexity in a weak-measure i.e. 
	$$
	\max\{\min\{\text{acceptance time for}~x~\text{on a nondeterministic path}\}: |x| = n, x \in L\},
	$$ 
	then such $\narm$ might be faster as some languages are in $O(\log \log n)$ steps, as shown later in Theorem \ref{the:nal_weak}, while a g-system can in each step increase the length only by a constant factor and needs at least $\Omega(\log n)$ steps.
	\item G-systems with deletion is more powerful than our model as they can recognize context-sensitive language in polynomial steps. The proof is based on the symbol duplication, which allows solving $PSPACE$ in polynomial steps \cite{Hartmanis74}.
	\item Sequential g-systems are weaker than the parallel one as they do not have the ability to duplicate a symbol. In fact, they are also weaker than our $\dal$ model. The fact that sequential g-systems work symbol-by-symbol with a limitation of copy-pasting on some positions shows that it is actually a restrictive automatic function steps.
\end{enumerate}

\noindent
We conclude our connection with the following evaluation. The essence of this is that sequential g-systems do not allow too powerful steps but are slower than our deterministic model, while parallel g-systems with deletion allow speed-ups which incorporate larger complexity-classes into the class of $NP$ in polynomial deterministic steps. Even the $\epsilon$-free g-system is comparable to our nondeterministic model. So our model may somehow sit in the middle of the two by accommodating the adequate speed-up while avoiding unrealistic blow-up.

\subsection{Iterated Uniform Finite-State Transducers and its Open Problem} 

\noindent
Kutrib, Malcher, Mereghetti, and Palano \cite{Kutrib20} introduced Iterated Uniform Finite-State Transducers as a model of computation preserving the polynomial time. This model uses length-preserving transducers as primitive operations and might be a real restriction as allowing only a subclass of $P$ for polynomial time recognition. As automatic functions can have output a constant longer than the input in each step, the overall model using automatic functions is not subjected to the length-limitation constraint and covers all polynomial-time decidable languages when recognizing languages in polynomially many steps, therefore more powerful. In a side note, the authors left it as an open problem whether their model (in arbitrary time) can cover the whole of context-sensitive languages. We prove the answer is positive by the following proof.

Roughly speaking, an Iterated Uniform Finite-State Transducer is a finite-state transducer that processes the input in multiple sweeps; in the first sweep, it reads the input followed by an endmarker and outputs a word; in subsequent sweeps, it reads the output word of the previous sweep and outputs a new word. Formally, a nondeterministic iterated uniform finite-state transducer is a system $T = \langle Q,\Sigma,\Delta,q_0,\triangleleft,\delta,F_+,F_- \rangle$, where $Q$ is the set of internal states, $\Sigma$ is the set of input symbols, $\Delta$ is the set of
output symbols, $q_0 \in Q$ is the initial state, $\triangleleft \in \Delta \setminus \Sigma$ is the endmarker, $F_+ \subseteq Q$ is the
set of accepting states, $F_- \subseteq Q \setminus F_+$ is the set of rejecting states and $\delta: Q \times (\Sigma \cup \Delta) \mapsto 2^{Q \times \Delta}$ is the transition function (which is total on $(Q \setminus (F_+ \cup F_-)) \times (\Sigma \cup \Delta)$);
the endmarker $\triangleleft$ is output only if it is read. (One observes from this definition that any word emitted by an Iterated Uniform Finite-State Transducer has the same length as the input word.)  $T$ halts whenever the transition function is undefined or $T$ enters an accept
or reject state at the end of a sweep. A computation of the nondeterministic iterated uniform finite-state transducer $T$ on input $w \in \Sigma^*$ is a sequence of words $w_1,\ldots,w_i,w_{i+1},\ldots$ such that $w_1 \in T(w\triangleleft)$ and $w_{i+1}
\in T(w_i)$; the computation halts if there is an $r \geq 1$ such that $T$ halts on $w_r$.  The set of possible words output by $T$ in a complete sweep on input $w \in (\Sigma \cup \Delta)^*$ is denoted by $T(w)$.  The input word $w$ is accepted by $T$ iff all computations on $w$ halt and at least one computation halts in an accepting state; $w$ is rejected by $T$ iff all computations on $w$
halt and none of the computations halt in an accepted state.

The proof is quite similar to the proof of showing context-sensitive languages are in nondeterministic linear space. The essential idea
is to construct a transducer that reads and outputs a convolution of three components: first, the input word (whose membership in the given context-sensitive language is to be tested); second; a ``work-tape'' for simulating context-sensitive grammar derivations; third, a counter for measuring the length of a derivation. At the beginning, let us fix a constant-sized ordered set of digits (in the alphabet) for counting e.g. $\{0,1,\cdots,9\}$, such that this set has size equal to the total number of terminals and non-terminals plus one (this is to ensure that whenever there is a repetition-free derivation of the word, the process does not run out of the steps to be counted, but can complete the derivation). Here, $9$ acts as the largest digit but it does not mean that there are exactly $10$ digits. The input/output alphabet of the transducer consists of all convolutions of the form $\icol{a\\b\\c}$, where $a$ (resp.~$b$) either represents a letter in the alphabet of the input grammar or is a blank symbol, and $c$ is either a digit in the counting alphabet or a blank symbol.  The first run of the transducer translates the input to a same-length-word in the tape alphabet which is a convolution of the input, the counter ``$00\cdots 0$'' and the start symbol $S$ (using blank symbols as appropriate to ensure that all three components are of equal length). 
Now for any given $n$, one may determine from the input grammar the maximum length of a repetition-free derivation of a word of length $n$. The counter starts from ``$0000 \cdots$'' and goes increasingly to ``$9999 \cdots$'' in some sufficiently large alphabet with a word as long as the input and forces a reject, if the input has not been derived in the corresponding maximum possible number of steps in a repetition-free derivation. 
If a derivation of a symbol $A$ becomes, under context, a word $w$ of length at least $1$, the transducer will nondeterministically guess the distribution of its symbols over the position of $A$ so that they go into the target position. 
If there is no non-terminal in the word, we just simply increment the counter. The machine will accept if the generated word and the given word are equal (independently of the counter number). The machine will reject if the last counter (``$9999 \cdots$'') is reached, thus if there is any incorrect guess then the derivation will be stuck and eventually be rejected.  

Finally, we remark that every automatic function which maps inputs to outputs of the same length can be computed by a transducer which is length-preserving.  The reason is that an automatic function is verified by a synchronous nfa which reads one symbol of each the input and the output per cycle; now turning the ``reading of the symbol'' into ``writing of a symbol'' where the successor state depends in the same way as before on the two symbols processed (the input and the output symbol), turns the nfa into a transducer which in each cycle reads one input symbol and writes one output symbol. At the input-end-symbol, this nfa indicates that the sweep was successful iff the nfa on the sequence of all the inputs processed and outputs generated would say that the (input,output)-pair is in the automatic relation. This way one can translate the nfa recognising the automatic relation into a transducer given as a Mealy machine which computes a multivalued function of the possible next steps in the derivation from the given step (with the corresponding maintenance of the counter and comparison with the input-word). 

\section{Key Results in DARM}
\label{sec:darm}

\noindent
As we can translate a $\darm$ to an equivalent deterministic Turing machine which runs linear-time longer, theoretically it may be --- but not always --- possible for some languages to be recognized by $\darm$ in up to linear-time fewer steps than an established algorithm for such language in standard Turing machine steps. This section highlights several key results of those interesting languages in $\dal$ class.

\subsection{Deterministic Polynomial ARM and Lower Bound Technique}
\label{subsec:darm_key}

\noindent
We start with the following theorem which is a direct consequence of authors' result in another work \cite{Gao21}.

\begin{theorem} \label{the:darm_regular}
	$\dal[1] = Regular$.
\end{theorem}

\begin{proof}
	Let $P_k$ be the class of languages decided by a $k$-step $\darm$ program. We will show that languages in $P_k$ are regular over constant $k$ by induction. For $k=1$, any language $L \in P_1$ must be regular as the automatic function step is a Deterministic Finite Automaton (DFA) recognizing $L$. Now assume that all languages in $P_k$ are regular. For any language $L \in P_{k+1}$ and its $\darm ~ M = \langle \Gamma,\Sigma,G \rangle$, let $g_k$ be the content (parameters) of $G$ after $k$-th step. One can construct a Nondeterministic Finite Automaton (NFA) guessing $g_k$ then verify with the automatic function step (i.e. a DFA) on the $(k+1)$-th step whether to accept or reject the string. This whole process can be replaced by another DFA of exponential-larger size by processing the convolution of the content of $G$ generated from the first step to the $(k-1)$-th step. Note that the program now consists of $k$ steps and the last step is done by a DFA therefore $L \in P_{k+1}$ is regular. 
\end{proof}

\noindent
Mentioned briefly in Subsection \ref{subsec:arm_turing} and \ref{subsec:transducer_pspace}, we can also prove the following equality.

\begin{theorem} \label{the:dal_ptime}
	$\dal[poly(n)] = P$.
\end{theorem}

\begin{proof}
	To show $\dal[poly(n)] \subseteq P$, note that in a single step of an automatic function the space used grows by a constant, thus each step update of the registers can be done in polynomial time in the length of the register string. For $P \subseteq \dal[poly(n)]$, we refer to the result from Hartmanis and Simon \cite{Hartmanis74} where $P$ can be covered by polynomial many steps of basic operations which are in fact automatic functions without any use of indirect addressing.
\end{proof}

\noindent
We also give the definition of a crossing sequence by Hennie \cite{Hennie65} as a tool to compute a lower bound on the computational complexity of a language and then connect it in the context of the execution of automatic relations in $\darm$.

\begin{definition}[Crossing Sequence, Hennie \cite{Hennie65}]
	Given a deterministic one-tape Turing Machine and an input, a crossing sequence of the $i$-th cell is a sequence of states whenever the head of Turing Machine crosses between the $i$-th and $i+1$-th cells along the computation. Informally, a crossing sequence of the $i$-th cell describes how the Turing Machine carries the information between the left and right cells separated by the border of $i$-th and $i+1$-th cell. The total computation time is the sum of the lengths of all crossing sequences.
\end{definition}

\noindent
The connection of crossing sequences to the execution of automatic relations in a $\darm~ M = \langle \Gamma,\Sigma,G \rangle$ is the following. Each step in the execution of $G$ can be simulated by a double-pass of a deterministic one-tape Turing machine. The machine passes first forward and then backward over the content of the register to replace the convolution (of $r$ and $r'$, where $r'$ is the content of $r$ after the execution of $G$) by their updated versions in all cells. As the state may change at the end of each pass, hence the crossing sequence becomes only at most two symbols (states) longer in one step update. In general, a program line number update in $\darm$ can also be simulated similarly. Thus, the length of a $\dal[f(n)]$-computation gives a crossing sequence of length $\bigO(f(n))$. So a lower bound on the crossing sequence length gives a lower bound on the $\dal$-computation too.

\subsection{Context-Free Languages and Boolean Languages in DARM}

\noindent
We first mention the CYK algorithm, discovered independently by Cocke, Kasami, and Younger \cite{Aho72,Kasami66,Younger67}: a parsing algorithm for context-free languages based on a bottom-up approach of dynamic programming. Given a context-free language $L$ in its grammar $G = \grammar$ in Chomsky Normal Form, the algorithm tries to parse an input string $s = s_1 \cdots s_n$ with length $n$ using CYK triangular matrix. Here is an example of the CYK triangular matrix for $|s| = 4$.

\begin{center}
	\begin{tabular}{*{9}{c}}
		&&&&$\alpha_{1,4}$&&&\\
		&&&$\alpha_{1,3}$&&$\alpha_{2,4}$&&\\
		&&$\alpha_{1,2}$&&$\alpha_{2,3}$&&$\alpha_{3,4}$&\\
		&$\alpha_{1,1}$&&$\alpha_{2,2}$&&$\alpha_{3,3}$&&$\alpha_{4,4}$\\
		\hline
		$s = $&$s_1$&&$s_2$&&$s_3$&&$s_4$\\
		\hline
	\end{tabular}
\end{center}

\noindent
The variables $\alpha_{i,j}$ will be the set of all non-terminals generating substring $s_i \cdots s_j$ i.e. $\alpha_{i,j}=\{A ~ | ~ A \in N, ~ A \Rightarrow^* s_{i,j}\}$. The algorithm starts with the bottom row computing $\alpha_{i,i} = \{A ~ | ~ A \rightarrow s_{i,i} \in P\}$. Iteratively, $\alpha_{i,j}$ can be efficiently computed using its (lower-)left diagonal: $\alpha_{i,k}$ where $k < j$; and (lower-)right diagonal: $\alpha_{k,j}$ where $k > i$ . This is captured by the matrix such that $\alpha_{i,j} = \bigcup_{i \leq k < j} ~ \alpha_{i,k} \times \alpha_{k+1,j}$ where $X \times Y = \{A ~ | ~ B \in X, C \in Y, ~ A \rightarrow BC \in P\}$. Finally, $s \in L$ iff $S \in \alpha_{1,n}$. The algorithm clearly runs in $\bigO(n^3)$ in modern RAM model and translates to an $\bigO(n^4)$-algorithm on a one-tape deterministic Turing machine. We show that the problem needs less number of steps in $\darm$.

\begin{theorem} \label{the:cfg_dal_n2}
	Given a context-free language $L$ in its grammar $G = \grammar$ in Chomsky Normal Form, recognizing an input with length $n$ in $L$ can be done in $DAL[n^2]$.
\end{theorem}

\begin{proof}[Proof Sketch]
	The idea is similar to using 1-dimensional arrays in solving context-free languages in quadratic steps by Kosaraju \cite{Kosaraju75}. We simulate CYK algorithm but for each variables $\alpha_{i,j}$ we only need constant steps to compute $\bigcup_{i \leq k < j} ~ \alpha_{i,k} \times \alpha_{k+1,j}$ which is done in parallel. We present the $\darm$ algorithm as follows.
	
	For $0 \leq k < n$, let the $k$-th layer of CYK triangular matrix be the $(k + 1)$-th row from the bottom i.e. the set of $\alpha_{i,i+k}$ for $1 \leq i \leq n-k$. Assign two registers $R_l$ and $R_r$ in $\darm$ to be responsible of left and right diagonals respectively. The $\darm$ algorithm will run in $n$ iterations and process layer-by-layer from $0$-th to $(k-1)$-th. When it finished processing the $k$-th layer, the algorithm guarantees the following properties:
	
	\begin{enumerate}
		\item Started with some padding, $R_l$ will be a concatenation of all left diagonals up to the $k$-th layer (i.e. left-diagonals of $\alpha_{i,j}$ for $i \leq j \leq \min\{n,i+k\}$) written from left-to right and separated by some separator.
		\item $R_r$ will be a concatenation of all right diagonals up to the $k$-th layer (i.e. right-diagonals of $\alpha_{1,i}$ for $1 \leq i \leq k$ and $\alpha_{i,j}$ for $i \leq j \leq \min\{n,i+k\}$) written from left-to right and separated by some separator, followed by some padding.
		\item The separators in $R_l$ and $R_r$ are lined up in such a way that $\alpha_{1,1}$ in $R_l$ is matched with $\alpha_{2,2+k}$ in $R_r$. In other words, full ``blocks'' of diagonals are aligned to compute $\bigcup_{i \leq k < j} ~ \alpha_{i,k} \times \alpha_{k+1,j}$ in the next iteration.
	\end{enumerate}
	
	\noindent
	Below is the illustration of above properties when processing a string of length $4$ for the first three iterations. A dashed line denote a separator symbol. Note that each of possible subsets of $N$ of $G$ can be denoted by a single symbol in $\Gamma$ of $\darm~M$ thus each $\alpha_{i,j}$ is represented as a single symbol.
	
	\begin{center}
		\begin{tabular}{ c c || c : c : c : c : c }
			\multirow{2}{*}{$0:$} & $R_l$ & $\#$ & $\alpha_{1,1}$ & $\alpha_{2,2}$ & $\alpha_{3,3}$ & $\alpha_{4,4}$ \\ 
			& $R_r$ & $\alpha_{1,1}$ & $\alpha_{2,2}$ & $\alpha_{3,3}$ & $\alpha_{4,4}$ & $\#$
		\end{tabular}
	\end{center}
	
	\begin{center}
		\begin{tabular}{ c c || c : c c : c c : c c : c c : c }
			\multirow{2}{*}{$1:$} & $R_l$ & $\#$ & $\#$ & $\#$ & $\alpha_{1,1}$ & $\alpha_{1,2}$ & $\alpha_{2,2}$ & $\alpha_{2,3}$ & $\alpha_{3,3}$ & $\alpha_{3,4}$ & $\alpha_{4,4}$ \\ 
			& $R_r$ & $\alpha_{1,1}$ & $\alpha_{1,2}$ & $\alpha_{2,2}$ & $\alpha_{2,3}$ & $\alpha_{3,3}$ & $\alpha_{3,4}$ & $\alpha_{4,4}$ & $\#$ & $\#$ & $\#$
		\end{tabular}
	\end{center}
	
	\begin{center}
		\small\begin{tabular}{ c c || c : c c : c c c : c c c : c c c : c c : c }
			\multirow{2}{*}{$2:$} & $R_l$ & $\#$ & $\#$ & $\#$ & $\#$ & $\#$ & $\#$ & $\alpha_{1,1}$ & $\alpha_{1,2}$ & $\alpha_{1,3}$ & $\alpha_{2,2}$ & $\alpha_{2,3}$ & $\alpha_{2,4}$ & $\alpha_{3,3}$ & $\alpha_{3,4}$ & $\alpha_{4,4}$ \\ 
			& $R_r$ & $\alpha_{1,1}$ & $\alpha_{1,2}$ & $\alpha_{2,2}$ & $\alpha_{1,3}$ & $\alpha_{2,3}$ & $\alpha_{3,3}$ & $\alpha_{2,4}$ & $\alpha_{3,4}$ & $\alpha_{4,4}$ & $\#$ & $\#$ & $\#$ & $\#$ & $\#$ & $\#$
		\end{tabular}
	\end{center}
	
	\noindent
	By the construction, the final iteration will eventually compute $\alpha_{1,n}$ thus the membership of $s$ is tested by checking whether $S \in \alpha_{1,n}$. To process $0$-th layer in particular, it can be done in linear steps by first transforming $s_i$ to $\alpha_{i,i}$, adding one padding symbol in the beginning of $R_l$ and at the end of $R_r$, then adding $n$ many times the separators.
	
	To process $k$-th layer from $(k-1)$-th layer, one can in a single automatic step compute all variables in the $k$-th layer by looking at the diagonals formed by the pairs in $R_l$ and $R_r$ --- ignoring incomplete blocks of diagonals. Note that $\bigcup_{i \leq k < j} ~ \alpha_{i,k} \times \alpha_{k+1,j}$ could be computed in that single step as there is an automaton recognizing the union of subsets of $N$ given that each $\alpha_{i,k}$ and $\alpha_{k+1,j}$ are already convoluted therefore the automaton can also compute a subset $\alpha_{i,k} \times \alpha_{k+1,j}$. After reading the separator, the variable $\alpha_{i,i+k}$ can be formed and convoluted to the separator (or stored in the third register $R_t$) and it continues to form the next variables. After all variables in the $k$-th layer is formed, the algorithm needs to put it on the correct positions: to the right of its left diagonal and to the left of its right diagonal. As its diagonals are directly before the separator, inserting it to the correct position can also be done in constant steps. Lastly, some paddings are added in the beginning of $R_l$ and at the end of $R_r$ such that $\alpha_{1,1}$ in $R_l$ is matched with $\alpha_{2,2+k}$ in $R_r$. Adding one padding needs one step, in addition, some separators might be added as well.
	
	Finally, to compute the total complexity time, one can observe that all operations are constant and could be reduced to how many times it inserts a symbol into $R_l$ and $R_r$. As the length of both will be equal to the size of CYK triangular matrix, therefore the algorithm runs in $\dal[n^2]$.
\end{proof}

\noindent
We can further extend the proof in Theorem \ref{the:cfg_dal_n2} for Boolean languages.

\begin{theorem} \label{the:bool_dal_n2}
	Given a Boolean language $L$ in its grammar $G = \grammar$ in Binary Normal Form, recognizing an input with length $n$ in $L$ can be done in $DAL[n^2]$.
\end{theorem}

\begin{proof}
	The CYK algorithm can be extended to parse Boolean grammars by Okhotin's work \cite{Okhotin13}. The algorithm is exactly the same with some modifications on computing $\alpha_{i,j}$ as now it allows conjunctions and negations. The algorithm also runs in $\bigO(n^3)$ in modern RAM model. In $\darm$, we can only modify the automaton when computing $\alpha_{i,j}$. Note that set operations such as intersection and complement can be done by introducing an additional Boolean flag and that will only multiply the number of states in the automaton by a factor of constant. Therefore the complexity will  still be $\bigO(n^2)$.
\end{proof}

\begin{remark}
	Note that the runtime complexity in Theorem \ref{the:cfg_dal_n2} and \ref{the:bool_dal_n2} still translates to an $\bigO(n^4)$-algorithms on a one-tape Turing machine due to the length of the register becoming $\bigO(n^2)$.
\end{remark}

\subsection{Multi-Sources Connectivity Problem in DARM}

\noindent
We define multi-source connectivity problem as the following.

\begin{definition}[Multi-Source Connectivity Problem]
	Given a directed graph $G = (V,E)$ with $n$ vertices $\{1,2,\cdots,n\}$, and some vertices as sources $S \subseteq V$, Multi-Source Connectivity Problem is to find all vertices reachable from at least one of the vertices in $S$.
\end{definition}

\noindent
We then define a string format of a directed graph $G = (V,E)$ with $n$ vertices $\{1,2,\cdots,n\}$ as the following. Note that `$|$' is a separator symbol.

\begin{center}
	\begin{tabular}{ | c | c c c c | c | c c c c | c | c | c c c }
		$v_1$ & $e_{1,1}$ & $e_{1,2}$ & $\cdots$ & $e_{1,n}$ & $v_2$ & $e_{2,1}$ & $e_{2,2}$ & $\cdots$ & $e_{2,n}$ & $\cdots$ & $v_n$ & $e_{n,1}$ & $\cdots$ & $e_{n,n}$
	\end{tabular}
\end{center}

\noindent
This format represents an adjacency matrix of $G$ which is written row-by-row. For each row, it starts with the vertex number $v_i = i$ followed by its neighbours: $e_{i,j} = j$ if $(i,j) \in E$, otherwise $e_{i,j} = 0$. Note that $v_i$ and $e_{i,j}$ are written in binary form thus each have length $\log n$. In this way, the string format of $G$ with $n$ vertices has total length of $\bigO(n^2 \log n)$.

As an input of multi-source connectivity problem, in addition to above string of $G$, we can also have a separate input for $S$: a string of length $n$ denoting the membership of each vertices in $S$ where the $i$-th symbol is $1$ iff $i \in S$ (else, $i$-th symbol will be $0$). Alternatively, if we just want the input to be a single string, all symbols of $v_i$ in the string of $G$ will be convoluted by a tick symbol $\checkmark$ (i.e. $| \binom{v_i}{\checkmark} |$) iff $i \in S$. Note that both input representations are equivalent and interchangeable with an additional $\bigO(n)$ simple automatic function steps to construct one from the other. To solve the multi-source connectivity problem, the algorithm must tick all reachable vertices by the similar convolution (or alternatively construct the membership string of length $n$ of reachable vertices). We then have the following results.

\begin{theorem} \label{the:connectivity}
	Multi-source connectivity problem of a directed graph with $n$ vertices can be solved in $DAL[n \log n]$.
\end{theorem}

\begin{proof}
	We use the input format with a single string (i.e. the vertices in $S$ is convoluted with tick symbol) and try to tick all reachable vertices. We give the following algorithm:
	
	\begin{enumerate}
		\item Find any ticked vertex $u$ having at least one out-going edge. If no such vertex exists, then terminate the algorithm.
		\item Find any $v$ such that $e_{u,v} = v$.
		\item Tick vertex $v$ (even if $v$ is already ticked).
		\item Remove all edges going to $v$ and repeat the algorithm.
	\end{enumerate}
	
	\noindent
	The algorithm uses the concept of flood-fill where it relays the connectivity from any visited (i.e. ticked) vertex to another vertex, shown by points $1$, $2$, and $3$. By removing all in-going edges to vertex $v$ in point $4$, $v$ will be guaranteed to be traversed at most once thus the algorithm will repeat at most $n$ times.
	
	Now we will calculate the time complexity. Point $1$ can trivially be done in a constant number of steps. Point $2$ is done similarly as the continuation. But in addition, the vertex $v$ is copied to another register in $O(\log~n)$ steps. The copy of $v$ then compared to each $v_i$ in parallel and the matched vertex number will be ticked in point $3$.  By ``comparing in parallel'' we mean here that the first character of $v$ is matched to the first character of every vertex, then the second character of $v$ is matched to the second character of every vertex, and so.  As each matching takes one round and $v$ is represented by $\log(n)$ bits, these comparisons need only $O(\log~n)$ steps. Point $4$ is done in the same fashion by comparing $v$ for each $e_{i,j}$ then set the matched ones to $0$, again in $O(\log~n)$ steps. As there are at most $n$ iterations, therefore the overall time complexity is $DAL[n \log n]$.
\end{proof}

\begin{corollary} \label{cor:bfs_dal_nlogn}
	Any graph problem with $n$ vertices whose solution relies on Breadth-First-Search (BFS) such as Shortest Path in Unweighted Graph, Cycle Detection, Bipartite Checking, and Eulerian Checking; can be solved in $DAL[n\log n]$.
\end{corollary}

\begin{proof}
	BFS can be done with some modifications of the algorithm of Theorem \ref{the:connectivity}. Instead of convoluting $v_i$ with a tick, $v_i$ will be convoluted with its distance, starting from $0$ on the source. Assign a register as a counter $cnt$, initialized with $0$. The modified algorithm is then the following:
	
	\begin{enumerate}
		\item Find any vertex $u$ with distance $= cnt$ and having at least one out-going edge. If no such vertex exists and $cnt < n$, then increment $cnt$. If no such vertex exists and $cnt = n$, then terminate the algorithm.
		\item Find any $v$ such that $e_{u,v} = v$.
		\item If the distance of vertex $v$ is not set yet, set it to $cnt + 1$.
		\item Remove all edges going to $v$ and repeat the algorithm.
	\end{enumerate}
	
	\noindent
	The proof of correctness is still the same as the previous algorithm but now "the ticks" are done in increasing order of the distance, instead of in an arbitrary order. The time complexity remains the same.
\end{proof}

\noindent
We also show that, if the graph is acyclic then a faster solution exists given that the graph is in a nice format.

\begin{definition}[Sorted-Directed-Acyclic-Graph]
	A Sorted-Directed-Acyclic-Graph (SDAG) is defined as a directed graph with no cycles and each edge $(i,j) \in E$ satisfies $i < j$. Note that every directed acyclic graph can be converted into its respective SDAG e.g. by performing a topological sort to relabel the vertex numbers.
\end{definition}

\begin{theorem} \label{the:sdag_dal_n}
	Multi-source connectivity problem of a SDAG with $n$ vertices can be solved in $DAL[n]$.
\end{theorem}


\begin{proof}
	We give the following algorithm in high-level. Assign two new registers denoted as $R_v$ and $R_e$ which has the same length as the input string. Initially, $R_v$ will have exactly a single marker symbol (denoted as $\uparrow$) located at the same position as the first symbol of $v_1$. $R_e$ instead will have exactly $n$ marker symbols located at the same positions as all the first symbol of $e_{i,1}$ for $1 \leq i \leq n$. All other symbols in $R_v$ and $R_e$ are blank.
	
	Now, for the next $n$ iterations, we maintain the markings in $R_v$ and $R_e$ such that on the $j$-th iteration: $R_v$ will have exactly one single marker at the first symbol of $v_j$ and $R_e$ will have exactly $n$ markers at the first symbol of each $e_{i,j}$ for $1 \leq i \leq n$. Here is the illustration of the movement of the markers. Note that, as we only mark the first symbol, the movement can be done in constant steps by moving the markers to their next spot in parallel.
	
	\begin{center}
		\small\begin{tabular}{c : c  || c | c c c c | c | c c c c | c | c | c c c c }
			Iteration &$G$ & $v_1$ & $e_{1,1}$ & $e_{1,2}$ & $\cdots$ & $e_{1,n}$ & $v_2$ & $e_{2,1}$ & $e_{2,2}$ & $\cdots$ & $e_{2,n}$ & $\cdots$ & $v_n$ & $e_{n,1}$ & $e_{n,2}$ & $\cdots$ & $e_{n,n}$ \\
			\hline
			\multirow{2}{*}{$1$}&$R_v$ & $\uparrow$ & & & $\cdots$ & & & & & $\cdots$ &  & $\cdots$ & &  & & & \\
			&$R_e$ &  & $\uparrow$ & & $\cdots$ & & & $\uparrow$ & & $\cdots$ &  & $\cdots$ & & $\uparrow$ & & & \\
			\hline
			\multirow{2}{*}{$2$}&$R_v$ & & &  & $\cdots$ & & $\uparrow$ & &  & $\cdots$ & & $\cdots$ & &  & & & \\
			&$R_e$ & & & $\uparrow$ & $\cdots$ & &  & & $\uparrow$ & $\cdots$ & & $\cdots$ & & & $\uparrow$ & &\\
			\hline
			$\cdots$ & & & &  & & &  & &  &  & & & & & & &\\
			\hline
			\multirow{2}{*}{$n$}&$R_v$ & & & & $\cdots$ &  & & & & $\cdots$ &  & $\cdots$ & $\uparrow$ & & & &  \\
			&$R_e$ & & & & $\cdots$ & $\uparrow$ & & & & $\cdots$ & $\uparrow$ & $\cdots$ &  & & & & $\uparrow$
		\end{tabular}
	\end{center}
	
	After the marking is maintained on the $j$-th iteration, our next job is to check whether $j$ is a reachable vertex. This can be done by checking if there exists any ticked (with $\checkmark$) vertex number $v_i$ that has an out-going edge to $j$ (i.e. $v_{i,j} = j$), if yes then tick vertex number $v_j$. Note that we only need constant steps as all relevant variables are already marked. After $n$ iterations, the program terminates and we show that all reachable vertex have been ticked.
	
	The proof of correctness is by an induction of this statement: after $j$-th iteration, all reachable vertices $i \leq j$ are all ticked. The base case of $j = 1$ is trivial because $1$ is reachable if and only if $1 \in S$ (which holds if and only if vertex $1$ is already ticked at the beginning). For the inductive step, suppose the statement is already correct for $j = k$. The $k+1$-th iteration checks whether vertex $k+1$ is reachable and that is true iff there exists a reachable vertex $i$ going to $k+1$ i.e. $v_{i,k+1}=k+1$. As the graph is SDAG, $i$ must be less than $k+1$. As we already correctly ticked the reachable vertices up to vertex number $k$, therefore the reachability of $k+1$ can be fully checked.
\end{proof}

\noindent
The upper bound of the multi-source connectivity problem is almost tight. In fact, it is match with the lower bound if the graph is a SDAG by the following theorem.

\begin{theorem} \label{the:daglower}
	The lower bound of multi-source connectivity of a directed graph with $n$ vertices is $\dal[\Omega(n)]$.
\end{theorem}

\begin{proof}
	Assume that $n$ is even so we can split the vertices into two equal groups: $\{1,2,\cdots,\frac{n}{2}\}$ and $\{\frac{n}{2}+1,\cdots,n\}$. Let us construct an input graph such that $(i,j) \in E$ iff $i + \frac{n}{2} = j$ and $S \subseteq \{1,2,\cdots,\frac{n}{2}\}$ i.e. the sources are the subset of the first group. By this construction, we reduce the problem to a duplication problem: the algorithm must tick the vertices in the second group such that they are equal to the ticks of the first group. It is also equivalent to transforming $w$ to $ww$ where $|w|=\frac{n}{2}$. Between both $w$, at least $\frac{n}{2}$ bits has to be transmitted as it carries the copy of $w$. Thus by its crossing sequence, the lower bound will be $\Omega(n)$-$\dal$-computation.
\end{proof}

\subsection{Sorting in DARM}

\noindent
We also have a nice result in sorting where the solution is optimal.

\begin{theorem} \label{the:sortdal}
	Given $n$ many $m$-digits numbers as an input, sorting those $n$ numbers can be solved in $DAL[nm]$. Moreover, it is also matched with the lower bound.
\end{theorem}

\begin{proof}
	We can directly implement parallel neighbour-sort by Habermann \cite{Habermann72}. The algorithm runs in $n$ rounds. In odd-numbered round, $\lfloor n/2 \rfloor$ pairs are set: $(1,2),(3,4),\cdots$; while in even-numbered round, $\lfloor (n-1)/2 \rfloor$ pairs are set: $(2,3),(4,5),\cdots$. These pairs denote the two index numbers to be compared and those two numbers will be swapped if the former is larger. In a single round, all comparisons and swappings can be done in parallel and need $\bigO(m)$ steps since each number has $m$ digits. Thus, the total complexity time is $DAL[nm]$.
	
	For lower bound, an input can be generated by putting $n/2$ large numbers in the first half and smaller numbers in the rest last half of the input. To sort this requires transferring all numbers in the first half to the last half and vice versa. Therefore, at least $nm/2$ bits of information have to be transmitted between the middle point. As crossing sequences are only determined up to a constant factor, thus the length of the crossing sequence between both halves will be in length $\Omega(nm)$, hence the lower bound is attained.
\end{proof}

\section{Key Results in NARM}
\label{sec:narm}

\noindent
In the same spirit of $\darm$, this section presents the relationship between $\narm$ and standard Turing machine complexity classes followed by membership of some interesting languages in $\bnal$ class.

\subsection{Nondeterministic Turing Machine, Regularity, and NONPALINDROME}

\noindent
We first start our result with the following theorem.

\begin{theorem}\label{the:bnal_nspace}
	$\bnal[f(n)] \subseteq \nspace[f(n)]\cap \ntime[(n+f(n))f(n)]$.
\end{theorem}

\begin{proof}
	Consider the space-time diagram of the register content in $\narm~M = \langle \Gamma,\Sigma,G \rangle$. Let $l$ be the maximum length of a string in the computation of $G$ automata until it halts. The input is padded with $\#$'s such that the length from first to last computation is always $l$. Let $r_i$ be the $i$-th symbol of the register.
	
	\begin{center}
		\begin{tabular}{|c || c | c | c | c |}
			\hline
			step & $r_1$ & $r_2$ & \ldots & $r_l$ \\
			\hline
			1 & $r_1^{(1)}$ & $r_2^{(1)}$ & \ldots & $r_l^{(1)}$ \\
			\hline
			2 & $r_1^{(2)}$ & $r_2^{(2)}$ & \ldots & $r_l^{(2)}$ \\
			\hline
			\ldots & \ldots & \ldots & \ldots & \ldots \\
			\hline
			$cf(n)$ & $x_1^{(cf(n))}$ & $r_2^{(cf(n))}$ & \ldots & $r_l^{(cf(n))}$ \\
			\hline
		\end{tabular}
	\end{center}
	
	\noindent
	Note that it suffices to keep track of the state of a deterministic finite automaton accepting the automatic relation step updates in order to simulate it. Hence, the Turing machine simply maintains a column of states. These states correspond to the states of the automata recognizing the automatic step updates. The Turing machine then guesses the contents column by column and simulates the automata to process these contents. At the end of all $l$ columns, it then verifies that all step updates are valid. The total amount of space used is simply $\bigO(f(n))$. Hence, we have $\bnal[f(n)] \subseteq \nspace[f(n)]$.
	
	Secondly, in each bounded automatic relation step, the length of the string is allowed to increase by at most a constant $c$. Hence, the length of the string after $\bigO(f(n))$ steps is bounded by $\bigO(n + f(n))$. By Theorem \ref{the:turing_simulation}, the execution of automatic relations can be simulated by a nondeterministic Turing Machine in $\bigO((n+f(n))f(n))$ steps. Hence, $\bnal[f(n)] \subseteq \ntime[(n+f(n))f(n)]$.
\end{proof}

\begin{corollary}
	$\bnal[poly(n)] = \np$.
\end{corollary}

\begin{corollary}
	For $f(n) = \Omega(n)$, under the assumption that $\ntime[(n+f(n))f(n)]\not\supseteq \nspace[f(n$ $)],\bnal[f(n)] \subset \nspace[f(n)]$.
\end{corollary}

\noindent
Next, we show several results on a $\narm$ with a smaller number of computations. Note that an automatic relation step is an NFA which can be constructed as a DFA but with exponential-larger size. Therefore, we can directly use the proof in Theorem \ref{the:darm_regular} to show that $\bnal[1] = \dal[1] = Regular$. However, a better bound can be achieved.

\begin{theorem}\label{the:nal_weak}
	$\bnal[o(\log n)] = Regular$, moreover if NARM measures the time complexity in a weak-measure i.e. $\max\{\min\{\text{acceptance time for}~x~\text{on a nondeterministic path}\}: |x| = n, x \in L\}$ then $\bnal[o(\log \log n)] = Regular$.
\end{theorem}

\begin{proof}
	In Subsection \ref{subsec:darm_key} we show that the length of a $\dal[f(n)]$-computation gives a crossing sequence of length $\bigO(f(n))$. This result can be extended to $\narm$ in a similar manner as a bounded automatic relation can be simulated by a double-pass of a nondeterministic one-tape Turing machine. Thus an $f(n)$-steps NARM produces an $f(n)$ crossing sequence length accepting computation on the best possible run. Szepietowski \cite{Szepietowski94} and Pighizzini \cite{Pighizzini09} show that if a nondeterministic one-tape Turing machine uses a crossing sequences of length at most $o(\log n)$ then it accepts regular language. This result allows to translate that $\narm$ in $o(\log n)$ steps accept regular language too.
	
	Moreover, Pighizzini \cite{Pighizzini09} also shows that in a weak-measure where the acceptance of a string is only measured to its shortest computation path, then the previous bound goes down to $o(\log \log n)$. Their proof is based on accepting a unary string of length $n$ where the smallest integer non dividing $n$ is not a power of $2$, by guessing $s$ and $t$ such that $2^s < t < 2^{s+1}$, $n \mod 2^s = 0$, and $n \mod t \neq 0$. This translates well for $\narm$ in a weak-measure by guessing a string where every $t$-th symbol is a $2$ and in between are symbols $1$, for some non-negative integer $t$ --- allowing $\narm$ to check the divisibility of $t$.
\end{proof}

\noindent
The next lemma nicely follows the hierarchy that in $O(\log n)$ steps, $\narm$ may recognize a non-regular language.

\begin{definition}[PALINDROME and NONPALINDROME]
	PALINDROME = $\{x \in \{0,1\}^*: \forall i, x_i$ $= x_{|x|-i+1} \}$ while NONPALINDROME = $\{0,1\}^* \setminus PALINDROME$.
\end{definition}

\begin{theorem}[B\={a}rzdiņ\v{s} \cite{Barzdin65}, Hennie \cite{Hennie65}, Rabin \cite{Rabin63}]\label{the:palindrome}
	PALINDROME cannot be accepted on a nondeterministic one-tape Turing Machine in $o(n^2)$ steps.
\end{theorem}

\noindent
The subsequent corollary follows from Theorems \ref{the:bnal_nspace} and \ref{the:palindrome}.

\begin{corollary}\label{cor:nonpalimnalon}
PALINDROME $\notin \bnal[o(n)]$.
\end{corollary}

\begin{lemma}\label{the:nonpalindrome}
	NONPALINDROME $\in \bnal[\log n]$.
\end{lemma}

\begin{proof}
	The idea is the following. After $r_1$ reads an input $w$ where $|w|=n$, $r_2$ guesses (i.e. nondeterministically chooses) a string $0^i \cdot 3 \cdot 1^j \cdot 3 \cdot 2^i$ where $2i + j + 2 = n$ such that the symbols of $w$ at positions $i+1$ and $n-i$ are different, thus $w$ is nonpalindrome. As $r_2$ can only guess in form of $0^*31^*32^*$ with an automatic relation, our task is then to verify that the length is exactly $n$ and the number of symbol `0' and `2' are equal. Verifying that the length is exactly $n$ is trivial by comparing $r_1$ and $r_2$ (reject if one string is longer/shorter). To verify that the occurrences of `0's and `2's are the same, while there exists symbol `0' or `2', repeat the following rounds:
	
	\begin{enumerate}
		\item If the parities of `0's and `2's are different, then reject.
		\item Else, modify the $1,3,5,\cdots$-th occurrences of `0' to `1', and modify the $1,3,5,\cdots$-th occurrences of `2' to `1'.
	\end{enumerate}
	
	\noindent
	If the iteration terminates without reject, both digits are disappearing at the same time and proven to be always equal. The final step is then to verify that the symbols on $r_1$ in the same positions of both `3's on $r_2$ are different, thus a nonpalindrome. Note that both the number of `0's and `2's are halved in each round and therefore the runtime of the verification is logarithmic in the length of the input.
\end{proof}

\noindent
Finally, we wrap up this subsection with the following result.

\begin{theorem}\label{the:bnal_logn}
	$\bnal[\log n] \subset \nl \subseteq P$. 
\end{theorem}

\begin{proof}
 Combining Corollary \ref{cor:nonpalimnalon} and Lemma \ref{the:nonpalindrome}, $\bnal[\log n]$ is not closed under complement, whereas $\nl$ is closed under complement \cite{Imm88}\cite{Sze87}. From Theorem \ref{the:bnal_nspace}, we also have $\bnal[\log n]\subseteq \nl$. Hence, $\bnal[\log n] \subset \nl$.
\end{proof}

\subsection{3SAT and Context-free Languages in NARM}

\noindent
We employ the following binary encoding of $\sat$ as an input.

\begin{definition}[Binary Encoding of $\sat$]
	Consider a binary encoding of $\sat$ where a variable $x_i$ is represented by a binary number $i$. We use $s_i \in \{+,-\}$ to represent whether the $i$-th literal is positive or negative such that a clause $(x_i \lor \lnot x_j \lor \lnot x_k)$ is represented by $i+~|~j-~|~k-$. Different clauses are separated by symbol ``$\&$''.
\end{definition}

\noindent
Let $n$ be the length of the formula (and not the number of variables). We can show a nice result that in unary encoding $\sat \in \bnal[n]$. However, using a similar technique, we can actually show that $\sat \in \bnal[\frac{n}{\log n}]$ if it is in binary encoding. The catch is that an additional small constraint has to be imposed: if there are $k$ variables then the first $k$ binary numbers (i.e. $1$ to $k$) must be used as their namings and are in the same length $O(\log k)$. The following is the proof.

\begin{theorem}
	In a nice format, $\sat \in \bnal[\frac{n}{\log n}]$ where $n$ is the length of the $\sat$ formula.
\end{theorem}

\begin{proof}
	After reading the input of formula $f$, assign a register $r_c$ as a counter. Register $r_c$ will copy the input $f$ but every variables in it is replaced by $0$. Then the following iteration is performed:
	
	\begin{enumerate}
		\item Increment every variables in $r_c$.
		\item Nondeterministically choose a truth value $T$ or $F$.
		\item For each variable in $f$ which is equal to the one in $r_c$, assign above truth value to it.
		\item If every variables in $f$ has been assigned, terminate the iteration; else repeat.
	\end{enumerate}
	
	\noindent
	The idea is that in the $i$-th iteration, the algorithm will assign the $i$-th variable a same truth value. Each operations in the iterations such as increment and assigning a truth value can be done in constant steps as it is possible update all variables both in $r_c$ and $f$ in parallel. After the iteration terminates, every variables has been assigned and a single step can be performed to check the satisfiability of the $f$.
	
	Note that the algorithm runs in $\bigO(k)$ steps where $k$ is the number of iterations. The number of iterations is equal to the number of variables if the first $k$ binary numbers are used as the naming. As each variable is written in the same length of $\bigO(\log k)$ therefore there must be at most $\bigO(\frac{n}{\log n})$ many variables, thus the runtime of the algorithm is $\bnal[\frac{n}{\log n}]$.
\end{proof}

\noindent
Beside $\sat$, we also have a result in context-free languages.

\begin{theorem}
	Context-free languages and their closure under union, intersection, concatenation, and Kleene star are in $\bnal[n]$; moreover, it is also optimal.
\end{theorem}

\begin{proof}
	The upper bound comes from a direct translation from grammars in Greibach normal form \cite{Greibach65}. Meanwhile, the lower bound follows from Corollary \ref{cor:nonpalimnalon}. 
As PALINDROME is a context-free language, 
there is no $\bnal[o(n)]$ algorithm to recognize the membership for context-free language.
\end{proof}

\section{Advanced Types of ARM}
\label{sec:other_arm}

\noindent
In this last section, we further extend the types of $\arm$ by allowing unboundedness of automatic relations, a booster operation to generate padding, and also a concept similar to Alternating Turing machines.

\subsection{Unbounded Nondeterministic Automatic Register Machine}
\label{subsec:unarm}

\noindent
This type of machine is similar to $\narm$ but it uses unbounded automatic relations instead.

\begin{definition}[Unbounded Nondeterministic Automatic Register Machine]
	An Unbounded Nondeterministic Automatic Register Machine ($\unarm$) is an Automatic Register Machine $M = \langle \Gamma,\Sigma,R,Op,P \rangle$ where $Op$ is restricted to unbounded automatic relations. Alternatively, a $\unarm$ is an Automatic Register Machine $M = \langle \Gamma,\Sigma,G \rangle$ where $G$ is an unbounded automatic relation.
\end{definition}

\begin{definition}[Unbounded Nondeterministic Automatic Register Machine Complexity]
	Let $f(n)$ be any function in the input size $n$. UNAL[$f(n)$] is the class of languages accepted by $\unarm$ in no more than $\bigO(f(n))$ steps.
\end{definition}

As unbounded automatic relation has the ability to guess a huge witness in one step, it may recognize a larger complexity class compared to $\narm$ as shown by the following result.

\begin{theorem}
	For $f(n) = \Omega(n)$, $\nspace[f(n)] \subseteq \unal[f(n)]$.
\end{theorem}

\begin{proof}		
	Consider a one-tape nondeterministic Turing Machine with $cf(n)$ cells for some constant $c$. Let $C_i$ denote the string on the tape of Turing Machine along with the head location (marked on the string) and state of the machine at the $i$-th step, padded with `\#' such that $|C_i| = cf(n)$. Let $|$ be a new symbol not in the alphabet of the Turing Machine. Now $\unarm$ nondeterministically guesses
	$$r = \begin{array}{c | c | c | c | c | c}
		C_1 & C_2 & \ldots & C_{i} & \ldots & C_{t-1} \\
		C_2 & C_3 & \ldots & C_{i+1} & \ldots & C_{t}
	\end{array}$$
	and verifies in one step that $C_1$ is the input (with starting state and the head location marked the leftmost symbol), $C_t$ is in accepting state, and $C_i$ to $C_{i+1}$ is a valid transition in each column separated by $|$.
	In the next $cf(n)$ steps, it verifies that $\forall i, C_i$ in the first row is the same as $C_i$ in the second row. Let $m = cf(n)$, expanding $C_i$ into $C_{i,1}C_{i,2}\ldots C_{i,m}$:
	$$r = \begin{array}{c | c | c | c}
		C_{1,1}C_{1,2}C_{1,3}\ldots C_{1,m} & C_{2,1}C_{2,2}C_{2,3}\ldots C_{2,m} & C_{3,1}C_{3,2}C_{3,3}\ldots C_{3,m} & \ldots \\
		C_{2,1}C_{2,2}C_{2,3}\ldots C_{2,m} & C_{3,1}C_{3,2}C_{3,3}\ldots C_{3,m} & C_{4,1}C_{4,2}C_{4,3}\ldots C_{3,m} & \ldots
	\end{array}$$
	
	\noindent
	In the $j$-th step, more specifically, it verifies that $\forall i$, $C_{i,j}$ from the second row of the $(i-1)$-th block is the same as $C_{i,j}$ from the first row of the $i$-th block. It also marks off $C_{i,j}$ with the marker symbol $@$ after checking. The $@$ will help the automatic relation to identify which is the $j$-th symbol:  the $j$-th symbol is exactly the first non $@$ symbol in each block.
	
	\noindent
	After one step:
	$$r = \begin{array}{c | c | c | c}
		\ @\ C_{1,2}C_{1,3}\ldots C_{1,m} & \ @\ C_{2,2}C_{2,3}\ldots C_{2,m} & \ @\ C_{3,2}C_{3,3}\ldots C_{3,m} & \ldots \\
		\ @\ C_{2,2}C_{2,3}\ldots C_{2,m} & \ @\ C_{3,2}C_{3,3}\ldots C_{3,m} & \ @\ C_{4,2}C_{4,3}\ldots C_{3,m} & \ldots
	\end{array}$$
	
	\noindent
	After two steps:
	$$r = \begin{array}{c | c | c | c}
		\ @\ \ @\ C_{1,3}\ldots C_{1,m} & \ @\ \ @\ C_{2,3}\ldots C_{2,m} & \ @\ \ @\ C_{3,3}\ldots C_{3,m} & \ldots \\
		\ @\ \ @\ C_{2,3}\ldots C_{2,m} & \ @\ \ @\ C_{3,3}\ldots C_{3,m} & \ @\ \ @\ C_{4,3}\ldots C_{3,m} & \ldots
	\end{array}$$
	
	\noindent
	Lastly, it also verifies all $C_i$ are having the same length, i.e. $|C_i| = m$, and it can be done with the help of a counter (for example, in another register). Hence, we have $\nspace[f(n)] \subseteq \unal[f(n)]$ for $f = \Omega(n)$. Note that the condition $f(n) = \Omega(n)$ is necessary as the simulation needs to verify that the correct input are used in the computation which takes $\Omega(n)$ steps. 
\end{proof}

\begin{corollary}\label{cor:unal_pspace}
	$\unal[poly(n)] = \pspace$.
\end{corollary}

\begin{remark}
	While it is clear that $\np \subseteq \pspace$, the equality/inequality part is still not known thus it is also unclear whether polynomial time models with bounded and unbounded automatic relations as operations are actually the same or different as a decision problem. As mentioned before, if both act as functions instead then the two concepts differ as a nondeterministic bounded automatic function can only have polynomial-sized output while the unbounded one can have exponential-sized output.
\end{remark}

\subsection{Polynomial-Size and Exponential-Size Padded Automatic Register Machines}

\noindent
Automatic functions and bounded automatic relations have a nice property that the length of the string only increases by a constant. As unboundedness introduces more power by its ability to guess a huge witness, we are curious whether the length of the string is what matters. What happens if we are still using automatic functions and bounded automatic relations yet we initially give the "working space" big enough? Here we define such models with a booster operation.

\begin{definition}[Polynomial-Size Padded Automatic Register Machine]
	
	A Polynomial-Size Padde-d Automatic Register Machine ($\parm$) is an ARM $M = \langle \Gamma,\Sigma,R,Op,P \rangle$ where instead of operation $read(R_i)$, there is an operation $read\_and\_boost(R_i,R_j,p)$ where $R_i,R_j \in R$ and $p$ is a polynomial. This operation will:
	
	\begin{enumerate}
		\item Assign the input string to register $R_i$.
		\item Let $n$ be the input length, then a padding string $0^*$ is generated in an adversarial way such that its length is at least $p(n)$ and then is assigned to $R_j$.
	\end{enumerate} 
\end{definition}

\begin{definition}[Exponential-Size Padded Automatic Register Machine]
	
	An Exponential-Size Pad-ded Automatic Register Machine ($\exparm$) is an ARM $M = \langle \Gamma,\Sigma,R,Op,P \rangle$ where instead of operation $read(R_i)$, there is an operation $read\_and\_boost(R_i,R_j,p,q)$ where $R_i,R_j \in R$ and $p,q$ are polynomials. This operation will:
	
	\begin{enumerate}
		\item Assign the input string to register $R_i$.
		\item Let $n$ be the input length, then a padding string $0^*$ is generated in an adversarial way such that its length is at least $2^{p(n)} \cdot q(n)$ and then is assigned to $R_j$.
	\end{enumerate} 
\end{definition}

\noindent
With $\parm$ model, a working space with polynomial-size larger than the original length of the input is generated and it can only be done once when reading the input. Any polynomial $p$ can be chosen by the designer of the algorithm but, due to the adversary nature of the booster step, it does not guarantee that the padding string has exactly $p(n)$ symbols --- it only guarantees that there are at least $p(n)$ symbols. The intention is that in $\bigO(1)$ steps only automatic functions and relations can be computed, i.e. it is not possible to have a function $h: x \rightarrow 0^{p|x|}$, thus a booster step is solely used to enlarge the working space. Note that one can however in $\bigO(\log n)$ steps trims the booster string to get the exact length of $p(n)$. This concept also applies for $\exparm$ model where the padding length will be at least $2^{p(n)} \cdot q(n)$ and all other properties can be chosen adversially.

\begin{remark}
	As any upper bound of $\parm$ and $\exparm$ are also welcome, one can modify the boosting operation to $read\_and\_boost(R_i,R_j,c)$ then just choose a constant $c$ such that the padding string has to have at least length $n^c$ in $\parm$ and $2^{n^c+c}$ in $\exparm$.
\end{remark}

\noindent
We then have the following definitions for the complexity.

\begin{definition}[Polynomial-Size Padded Deterministic Automatic Register Machine Complexity]
	Let $f(n)$ be any function in the input size $n$. $\pdal[f(n)]$ is the class of languages accepted by deterministic PARM in $\bigO(f(n))$. Here the PARM has to finish in $\bigO(f(n))$ steps no matter what length of the padding string a booster step generate.
\end{definition}

\begin{definition}[Exponential-Size Padded Deterministic Automatic Register Machine Complexity]
	Let $f(n)$ be any function in the input size $n$. $\expdal[f(n)]$ is the class of languages accepted by deterministic ExpARM in $\bigO(f(n))$. Here the ExpARM has to finish in $\bigO(f(n))$ steps no matter what length of the padding string a booster step generate.
\end{definition}

\noindent
It turns out allowing a large working space gives the machine more power as shown by the following result.

\begin{theorem}\label{the:qsat_expdal}
	The problem $QSAT$ is in $ExpDAL[n]$.
\end{theorem}

\begin{proof}
Given an input $x$ which is a QSAT-formula of length $n$, let $m$ be the number
of variables in it. One assumes the following coding and only in that case
the number returned by the algorithm needs to be correct:
\begin{enumerate}
	\item There is a prefix of the variables saying whether they are universal or
	existential quantified and the variables are ordered as numbers $m-1,
	m-2,\ldots,2,1,0$ in this prefix;
	\item Each variable name has $k$ binary digits and $m \leq 2^k < 2m$;
	\item The quantified part is in square brackets and within this,
	the variables are either prefixed by + or by - to indicate
	that they are non-negated or negated in the clause, respectively, and
	they have furthermore an underscore \verb|_| after the variable to
	record the value and they are separated by commas inside the clause
	and clauses are separated by semicolons;
	\item Every number $0,1,\ldots,m-1$ (in binary) occurs as a variable name
	somewhere in the inner part of the formula;
	\item The instance is redundance-free, that is, no clause $\alpha$ logically
	implies a clause $\beta$.
\end{enumerate}
For $m=3$, an example of a correct coding is the following:
$$
\exists 10 \forall 01 \exists 00 [+00\_,+01\_,+10\_;-01\_,-10\_]
$$
and it says that variable number 2 is existentially quantified followed by
variable 1 being universally quantified followed by variable number $0$
being existentially quantified such that one of the variables is positive
and one of the variables number 1 and number 2 is negative. The variable
order is in each quantification from top to bottom and that exactly the
variables are quantified which occur inside the formula.
If this is not the case then the algorithm can return any value.
Now the algorithm is as follows:

\begin{enumerate}
	\item Read input $x$ and let $n$ denote its length;
	\item One guesses (i.e. uses booster step) a string $y$ of length at least $2^{n+2} \cdot n$
	and in the next step one makes all symbols in $y$ to $3$;
	\item While the position of the second $3$ in $y$ is not at least two
	symbols after the end of $x$ do Begin
	Change in $y$ the second, fourth, sixth, $\ldots$ occurrence of $3$
	to $5$ End;
	\item Make the symbol after the last occurrence of $3$ to $4$ and delete
	everything afterwards (this can either increase the length of $y$ by
	$1$ or leave it as it is or make $y$ shorter);
	\item For all symbols $a$ between the two square brackets in $x$ (which might
	be read out by using a marker in the loop) and then for $a$=``.'' Do Begin
	For each occurrence of a $3$ in $y$ go to the first $5$ after it
	and replace it by $a$ End;
	\item Make a copy $z$ of $y$; Replace all $1$ in $z$ by $0$;
	\item If there are still some variable names which do not have an $s$ or $u$
	after them and either there is no position where the variable names in $y$
	and $z$ coincide and or there is a position where the variable name in
	$z$ consists only of $1$ then return the value $0$;
	\item If there is no variable name in $y$ which is not followed by either
	$s$ or $u$ then go to line 13;
	\item Go over $y$ and $z$ in parallel and with a bit $b$ having value $u$ at
	the beginning such that whenever one passes a $3$ in $y$ then one
	swaps $b$ from $u$ to $s$ or from $s$ to $u$
	and whenever one passes a variable name which coincides in $y$ and $z$
	then one replaces in $y$ the $\_$ after it by the value of $b$;
	\item Go over $y$ and change the second, fourth, sixth, $\ldots$ $3$ to
	$2$;
	\item Go over $z$ and increment all binary variable names by $1$ (an automatic
	function can do this);
	\item Go to line 7;
	\item In $y$, make the symbol after the second occurrence of a $3$ to $4$ (this
	$3$ sits after the first instance where all variables are evaluated
	to $s$) and discard all symbols after this $4$;
	\item Now replace in $y$ all $2$ by $3$;
	\item For each instance in $y$ (where an ``instance'' in $y$ refers to the segment between 
two consecutive occurrences of $3$), if every clause contains either a substring
	$+ \{0,1\}^* s$ or a substring $- \{0,1\}^* u$ then all clauses are
	satisfied and the first $5$ after the dot at the end of the
	instance is made an $s$ else the instance is not satisfied and the
	first $5$ after the dot at the end of the instance is made
	a $u$;
	\item Go backward through the quantifier prefix of $x$ (using some marker)
	until again $y$ has only two $3$s Do Begin
	\item For each subsequent two occurrences of $s,u$ after a dot do Begin \\
	if the quantifier is universal and both are $s$ then replace the
	first by $5$ and the second one by $s$ else \\
	if the quantifier is universal and at least one is $u$ then replace
	the first by $5$ and the second by $u$ else \\
	if the quantifier is existential and both are $u$ then replace the
	first by $5$ and the second by $u$ else \\
	if the quantifier is existential and at least one is $s$ then replace
	the first by $5$ and the second by $s$ End End;
	\item Now there is only one $s$ or $u$ left after a dot and this is the
	result of the formula ($s$ is satisfied and $u$ is unsatisfied).
\end{enumerate}

\noindent
Here are some explanations. First note that functions which are computed by
one-way Turing machines going either over one input alone or over several
inputs synchronously and doing changes on the spot are automatic and all
automatic functions in this algorithm are of this type; for line 11 the
Turing machine has to go backward for easily seen to be deterministic.
Actually nondeterministic Turing machines going forward can also do the
job. The characterisation \cite{Case}
says that the Turing machine processing the
variable $y$ (or $y$ and $z$ in parallel) for an automatic function can
change the direction constantly many times, but here the easier one-directional
pass is sufficient for all automatic functions; when comparing the variable
names in $y$ and $z$, the Turing machine goes over the word by accessing
in both variables the same position in each step.

\begin{itemize}
\item In line 2, $4n$ is a safe upper bound of the length of copy of the inner part of the
quantified variable plus some additional space and $n$ is a safe upper bound
on the number of variables, thus the length reserved in $y$ is sufficient
(one cannot fix the length exactly but only get it above the bound for
avoiding of coding).

\item Line 3 thins out symbols $3$ in $y$ by always converting every second to $5$
(what an automatic function can do) until the gap between the first $3$ (at the
beginning fo the word) and the second $3$ is longer than the length of $x$,
this are $O(\log n)$ many runs through the body of the loop which is implemented
by one automatic function.

\item Line 4 rectifies the end of the word and removes trailing $5$s and adds a $4$
as an end marker.

\item Line 5 copies in $O(n)$ rounds the inner part of the formula into $y$ starting
after each $3$ (there is enough space for this) by always for each symbol $a$
copied replacing the first $5$ after each $3$ by $a$. Each copying of a symbol
$a$ into all these places can be done by invoking one automatic function to
update $y$ and thus in overall $O(n)$ steps the copying can be done.

\item Line 6 copies $y$ into a new variable $z$ and makes in $z$ all variable names
to $00\ldots 0$ by making all $1$ to $0$.

\item Lines 7--12 are one big loop in order to set in $y$ all values of the variables
in a consistent manner with the additional property that all $2^m$ ways of
putting the variable values are covered in an instance.

\item Line 7 causes early termination of the algorithm in the case that the names
of the used variables are not contiguous.

\item Line 8 leaves the loop in the case that all variables are assigned $s$
(for satisfied or true) or $u$ (for unsatisfied or false).

\item Line 9 and line 10 make the values of the $\ell$-th variable to change
every $2^\ell$ instances, initially the zeroth variables changes from
every instance to the next between $u$ and $s$ and subsequently each
second, fourth, sixth, $\ldots$ $3$ becomes a $2$ to increase the
number of instances with the same value from $2^\ell$ to $2^{\ell+1}$
before processing the $\ell$+first variable.

\item Before going into the next iteration of the loop, all variable names
in $z$ are increased from $\ell$ to $\ell+1$.

\item Line 13 is just to make sure that exactly $2^m$ instances will survive
and these have the variable values $(u,\ldots,u,u)$, $(u,\ldots,u,s)$,
$(u,\ldots,s,u)$, $(u,\ldots,s,s)$, $\ldots$, $(s,\ldots,s,s)$ and a
one-way Turing machine going over the word $y$ can detect the position of the
instance where all variables are $s$ and change the symbol after the
$3$ after that instance to $4$ and discard the symbols after that;
thus this modification can be computed by an automatic function.

\item Line 14 is only a step needed to make the separation by two neighbouring
instances always by a 3 and not by a 3 or a 2. This eases the writing
of the algorithm, but has no further significance.

\item For each instance in $y$, a one-way Turing machine can check whether all
literals with value $s,u$ of the variable are set such that the instance
is satisfied, if so then the instance gets after its final dot the value $s$
else the instance gets after its final dot the value $u$. This is done in
one scan for all instances in $y$.

\item The loop in Lines 16 and 17  goes backwards over the quantifiers in $x$ and for
each quantifier, it identifies the still active values of blocks of
instances in the form $.s$ and $.u$ after the instance and it evaluates
the first and the second with the result in the second, the third and the
fourth with the result in the fourth, the fifth and the sixth with the
result in the sixth and so on in order to process the currently quantified
variable; in the $\ell$-th round $(\ell$ goes from $0$ to $m-1)$ this loop
treats the quantification of the $\ell$-th variable and takes the result
of two neighbouring blocks blocks of $2^{\ell}$ instances with the
$\ell$-th variable being $u$ and $s$ respectively from before and puts
the result into the entry for the second of these blocks while the entry
for the first of these blocks get erased. After doing this for all $m$
variables, one entry for the block of $2^m$ instances remains and this
is the result of this loop and value of the QSAT formula.

\item Line 18 returns this value. Note that the number of rounds in each of the
loops is $O(\log n)$ in the loop of Line 3, $O(n)$ for the loop copying
the formula in all applicable positions in Lines 4--5, number $m$ of
variables for the big loop in lines 7--12, number $m$ of variables for the
loop in lines 16--17. Thus the overall amount of operations in the evaluation
of the QSAT formula is $O(n)$; note that $m \leq n$.

\end{itemize}

\noindent
This completes the algorithm and its explanations for this proposition.
\end{proof}

\noindent
We also show the connection between $\exparm$ and $\unarm$ as the following.

\begin{theorem}
	$PSPACE = \expdal[poly(n)] = \unal[poly(n)]$.
\end{theorem}

\begin{proof}
We first prove $PSPACE \subseteq \expdal[poly(n)]$. Let a problem $R$ in $PSPACE$ be given which is many-one reduced by a function $f$ in polynomial time to QSAT where $f$ increases the size of an instance $n'$ to at most $p(n')$ for some polynomial $p$. Let an input $x'$ of $R$ be given. One generates in padding an input $y$
of length $2^{p(n') + 4} \cdot p(n')$ and then one translates
the instance $x'$ into $x=f(x')$ of length $n$ with $n \leq p(n')$
in $poly(n')$ steps. Afterwards one runs above $QSAT$ algorithm in $O(n)$
= $O(p(n'))$ steps to decide whether $x$ is in $QSAT$, without loss of
generality $f$ produces a formula which satisfies all the constraints
of the above algorithm so that the algorithm is correct.
Thus the algorithm is correct.

Next, we prove that $\expdal[poly(n)] \subseteq \unal[poly(n)]$:
One just replaces the step of producing the $y$ by exponential padding
having at least $2^{n^c+c}$ symbols (where $c$ is a sufficiently
large constant) by an unrestricted guessing of the $y$ and then one makes
all bits in $y$ to $1$ and then one verifies that $n^c+c$ many steps
of the algorithm
which converts in each step half of the $1$s to $0$s does not produce a
$0$ only word. Then the word is long enough and one can proceed with the
previous ExpDAL$[poly(n)]$ algorithm. By Corollary \ref{cor:unal_pspace}, we know that $\unal[poly(n)] = PSPACE$ therefore the equality is established.
\end{proof}

\noindent
We move on to the result in polynomial-size padding.

\begin{theorem}\label{the:pdal_polylog}
	Let $c$ be a constant. Solving $(c \log n)$-variable $\sat$ is in $PDAL[\polylog(n)]$.
\end{theorem}

\begin{proof}
Without loss of generality the input is given as a pair $(x,y')$ where
$y'$ has length $n^{c+2}$ and $x$ is a 3SAT formula (all variables existentially
quantified) written in the conventions of Theorem \ref{the:qsat_expdal}. In particular
the variables are written in binary and the instance $x$ is redundance-free.
Thus the instance $x$ has at most length $10 (c \log n \log \log n)^3+c'$
for some constant $c'$ and the factor $10$ is larger than the expected
factor $8$ to absorb constant add-ons per literal and the quantifier prefix.
If $x$ violates this length bound then one can return the value $0$, as
the input is not in adequate form.

Note that $n$ can be computed from $y'$ in $\log(n)$ rounds. The computation
and checking of the length-bound can be done in polylogarithmic many steps.
Furthermore, one can do a padding using that the $y$ guessed must have
$2^{c \log(n)}$ instances of length $|x|$ (which is polylogarithmic and thus
shorter than n)  so that $n^{c+2}$ is a save lower bound for the length
of the padding parameter. The algorithm of Theorem \ref{the:qsat_expdal} can now run with
$x$ and $y$ using $|y| \geq 2^m \cdot |x| \cdot 4$ and it needs
$O(|x|)$ steps where $|x| \leq 10 (\log n \cdot c)^3 + c'$. Thus the overall
amount of steps done is polylogarithmic in $n$.
\end{proof}

\noindent
Lastly, we prove that $\pdal[o(n)] \neq \dal[o(n)]$ if the following hypothesis is true.

\begin{hypothesis}[Exponential Time Hypothesis \cite{Impagliazzo01}]
	Exponential Time Hypothesis (ETH) states that solving an $h$-variable 3SAT instance needs time $\Omega(2^{\delta h})$ for some $\delta > 0$.
\end{hypothesis}

\begin{theorem}
	If ETH is true then $\pdal[o(n)] \neq \dal[o(n)]$.
\end{theorem}

\begin{proof}
	Let $h = c \log n$ such that $c \delta > 2$. It follows from ETH that solving the $h$-variable $\sat$ needs $\Omega(2^{c \delta \log n})$. As we set $c \delta > 2$ thus the lower bound will be $\Omega(n^2)$. Suppose there is an algorithm in $\darm$ solving it in $\dal[o(n)]$. By Theorem \ref{the:turing_simulation_deterministic}, such algorithm can be translated to a Turing machine algorithm which runs in $o(n^2)$ thus we get a contradiction. It is clear that the problem can be solved in $\pdal[o(n)]$ by Theorem \ref{the:pdal_polylog} therefore $\pdal[o(n)] \neq \dal[o(n)]$. 
\end{proof}

\subsection{Alternating Automatic Register Machine}

\noindent
Alternating Turing machine is formulated by Chandra, Kozen and Stockmeyer \cite{Chandra81} to generalize the concept of nondeterminism. We bring such notion of alternation to $\arm$ called Alternating Automatic Register Machine ($\aarm$) and have defined and investigated the model in another work \cite{Gao21}. We denote $\aarm$ complexity by $\baal$ and its polynomial-size padded version by $\paal$. The important findings of $\aarm$ are the following:

\begin{enumerate}
	\item $\bigcup_k \baal[n^k] = \pspace$.
	\item $\baal[1] = \textit{Regular}$, while there is already an $\np$-complete problem in $\baal[\log^* n]$.
	\item $\ph \subseteq \paal[\log^* n] \subseteq \pspace$.
\end{enumerate}

\section{Conclusion}

This paper introduced and studied several variants of register machines that employ 
automatic functions as primitive steps.  The first, known as a {\em Deterministic Automatic Register Machine} (\darm), is provably more computationally powerful than conventional register
machines.  For example, it was shown that any graph problem     
with $n$ vertices whose solution relies on Breadth-First-Search (BFS)
can be solved by a \darm\ in $\bigO(n\log n)$ time; moreover, the multi-source connectivity 
problem can be solved by a \darm\ in $\bigO(n)$ time.  Allowing non-determinism gives rise to
a more powerful type of register machine known as a 
{\em Non-Deterministic Automatic Register Machine} (\narm).  
It was shown that the class of context-free languages can be recognised
by a \narm\ in linear time, an improvement on the quadratic time 
complexity of this class for \darm s; furthermore, \sat\ can be recognised
by a \narm\ in $\bigO\left(\frac{n}{\log n}\right)$ time when the input formula is given in
a particular format, where $n$ is the length of the input formula. 
Finally, we introduced Automatic Register Machines
that are given more ``working space'' by means of a padding of the input
with a string of $0$'s; a {\em Polynomial-Size Padded Automatic Register Machine}
(\parm) is provided with a polynomial-size padding of the input
while an {\em Exponential-Size Padded Automatic Register Machine} (\exparm) 
is provided with an exponential-size padding.  It turned out that
an \exparm\ is just as powerful as an unbounded \narm\ (unbounded in the
sense that during each computation of an automatic relation, the output can be 
arbitrarily longer than the input), and both types of machines can recognise all 
languages in \pspace\ in polynomial time.

It may be natural to draw parallels between the models of computation
studied in the present work and other automata-based machines.
In this regard, we mention one such class of machines: cellular automata.
One can show that for any function $f$, $\bnal[f(n)]$ is characterised by  
a certain class of non-deterministic cellular-automata with $\bigO(n+f(n))$ cells.
In addition, both $\dal[f(n)]$ and $\bnal[f(n)]$ can be characterised by tiling systems,
though these characterisations are somewhat unnatural.
As such machine comparisons are not the main focus of this work,
we omit the details. 










\end{document}